\documentclass[aps,pra,10pt,a4paper,amsmath,amssymb,nofootinbib,notitlepage,tightenlines,longbibliography,superscriptaddress,twocolumn]{revtex4-1}
\usepackage[left=1.95cm,right=1.45cm,top=1.8cm,bottom=3cm]{geometry}

\usepackage{natbib}
\usepackage{hyperref}
\usepackage{cleveref}
\crefname{figure}{figure}{figures}

\usepackage{amsmath,amssymb}
\usepackage{bbm}
\usepackage{array}

\usepackage{enumerate}

\usepackage{graphicx}
\usepackage[font=small,justification=centerfirst]{caption}
\newcommand{\caphead}[1]{{\bf #1}}

\usepackage{amsthm}
\newcounter{tlc}

\newtheorem{theorem}[tlc]{Theorem}

\newcounter{lemmaN}

\newtheorem{lemma}[lemmaN]{Lemma}
\newtheorem{definition}{Definition}

\newcounter{assc}

\newtheorem{assumption}[assc]{Assumption}

\newcommand{\ket}[1]{\left | #1 \right\rangle}

\newcommand{\ketbra}[2]{|#1\left\rangle\right\langle #2 |}
\newcommand{\id}{\mathbbm{1}}

\DeclareMathOperator{\tr}{Tr}

\newcommand{\Prs}[2]{\mathrm{P}_{#1}\!\left(#2\right)}
\newcommand{\cPr}[2]{\mathrm{P}\!\left(#1 \,\middle|\, #2\right)}
\newcommand{\cPrs}[3]{\mathrm{P}_{\!#1}\!\left(#2 \,\middle|\, #3\right)}
\newcommand{\trans}[0]{^{\rm T}}

\newcommand{\corr}[1]{{C}\!\left(#1\right) }

\newcommand{\reals}{\mathbb{R}}
\newcommand{\nats}{\mathbb{N}}
\newcommand{\ints}{\mathbb{Z}}
\newcommand{\sphere}[1]{{S}^{#1}}

\newcommand{\SO}[1]{\mathrm{SO}\!\left(#1\right)}

\newcommand{\order}[1]{\mathcal{O}\!\left(#1\right)}

\makeatletter
\newcommand*{\balancecolsandclearpage}{%
  \close@column@grid
  \clearpage
  \twocolumngrid
}
\makeatother

\newcommand{\IQOQI}{Institute for Quantum Optics and Quantum Information,\\ Austrian Academy of Sciences, Boltzmanngasse 3, A-1090 Vienna, Austria}
\newcommand{\UoW}{Faculty of Physics, University of Vienna, Boltzmanngasse 5, A-1090 Vienna, Austria}
\newcommand{\Peri}{Perimeter Institute for Theoretical Physics, 31 Caroline Street North, Waterloo, ON N2L 2Y5, Canada}

\begin{document}

\title{
Semi-device-independent information processing\\with spatiotemporal degrees of freedom
}

\author{Andrew J.\ P.\ Garner}
\affiliation{\IQOQI}

\author{Marius Krumm}
\affiliation{\IQOQI}
\affiliation{\UoW}

\author{Markus P.\ M\"{u}ller}
\affiliation{\IQOQI}
\affiliation{\Peri}

\begin{abstract}
Nonlocality, as demonstrated by the violation of Bell inequalities, 
 enables device-independent cryptographic tasks that do not require users to trust their apparatus. 
In this article, we consider devices whose inputs are spatiotemporal degrees of freedom, e.g.\ orientations or time durations.
Without assuming the validity of quantum theory, 
 we prove that the devices' statistical response must respect their input's symmetries,
 with profound foundational and technological implications.
We exactly characterize the bipartite binary quantum correlations in terms of local symmetries,
 indicating a fundamental relation between spacetime and quantum theory. 
For Bell experiments characterized by two input angles, 
 we show that the correlations are accounted for by a local hidden variable model if they contain enough noise, 
 but conversely must be nonlocal if they are pure enough. 
This allows us to construct a  ``Bell witness'' that certifies nonlocality with fewer measurements than possible without such spatiotemporal symmetries, suggesting a new class of semi-device-independent protocols for quantum technologies. 
\end{abstract}

\date{February 3, 2020}

\maketitle

\section{Introduction}
\enlargethispage{0.85\baselineskip}
Quantum theory radically challenges our classical intuitions. 
A famous example is provided by the violation of Bell inequalities~\cite{EinsteinPR35,Bell64,ClauserHSH69,AspectGR82,Hensen15,BrunnerCPSW14},
 demonstrating that local hidden variable models are inadequate to account for all observable correlations in quantum theory.
While this so-called {\em nonlocality} was initially of foundational concern,
 it transpires to have a very powerful practical use: it enables {\em device-independent} protocols in quantum information theory (e.g.~\cite{MayersY98,BarrettHK05,ColbeckR12,VaziraniV14}). 
In this paradigm, one can perform certain tasks (e.g.\ cryptography) without trusting one's apparatus, or even necessarily assuming the full formalism of quantum mechanics.
These protocols rely on the readily believable {\em no-signalling} constraint, which forbids the instantaneous transmission of information between sufficiently distant laboratories.
Since this constraint originates in special relativity, it may be thought of as a property of spacetime itself.

A pillar of the device-independent formalism is its abstract {\em black box} description: 
 experimental devices are fully characterized by probability tables of outputs given a supplied input (\cref{fig:BellExpt}a).
In this article, 
 we supplement these inputs with physical structure,
 and adopt a semi-device-independent approach that makes no assumptions about the inner workings of the devices, or the physical theories governing them (i.e.\ quantum or otherwise),
 but assumes that their ensemble statistics can be characterized by a finite number of parameters.
Specifically, we consider when inputs are {\em spatiotemporal} degrees of freedom, e.g.\ some orientation in space or duration of time.
This includes, for example, the bias of a magnetic field, duration of a Rabi pulse, or angle of a polarizer (\cref{fig:BellExpt}b).
Spatiotemporal degrees of freedom bring with them a symmetry structure,
 which can be mathematically described using Lie group theory.

\begin{figure}[tbh]
\includegraphics[width=0.275\textwidth]{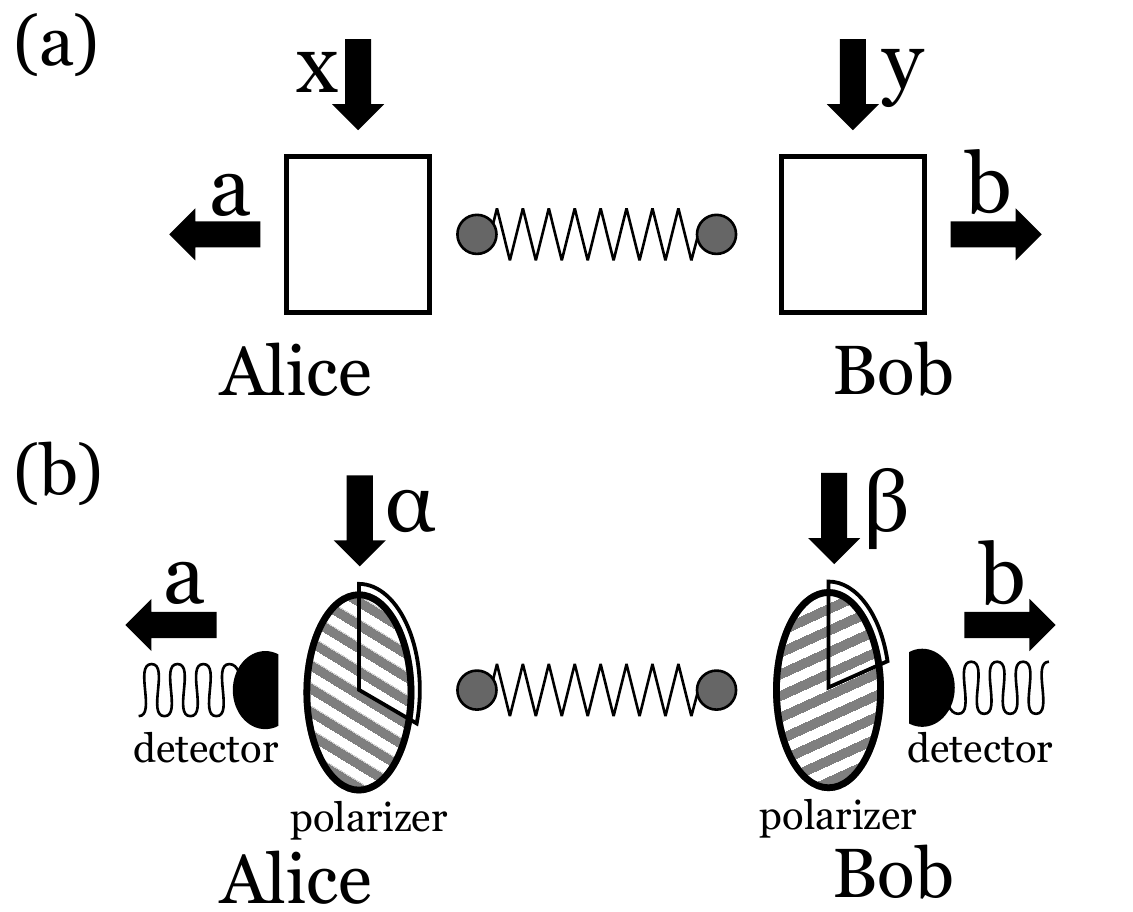}
\caption{
\label{fig:BellExpt}
\caphead{Bell scenario: abstract vs.\ spatiotemporal inputs.} 
Spatially-separate Alice and Bob independently choose measurement settings $x$, $y$ and receive some outputs $a$, $b$, yielding the joint conditional probability distribution $\cPr{a,b}{x,y}$.
(a) In the usual black box formalism, the inputs $x$ and $y$ are abstract labels. 
(b) Here, we consider the physical situation where the inputs are spatiotemporal degrees of freedom (e.g.\ angles $x=\alpha$, $y=\beta$ of polarizers).
}
\vspace*{-1.925em}
\end{figure}

In this article, we introduce a general framework for spatiotemporal black boxes.
We prove that the probability tables associated with spatiotemporal inputs must encode a linear representation of the corresponding symmetry groups (\cref{sec:BlackBox}). 
We demonstrate the
 power of this approach with two examples in Bell test scenarios:
First, if each laboratory controls a single angle (\cref{sec:SO2}), 
 we find---independently of the theory---that the response to rotations can in some cases certify the existence of a local hidden-variable model, or the violation of a Bell inequality. 
Consequently, we present a novel protocol for witnessing nonlocality,
 similar in spirit to \cite{TuraASVLA14,SchmiedBAFSTS16},
 but without prerequiring the validity of quantum theory.
Secondly, we consider when both inputs are chosen via rotations in $d$-dimensional space. 
We show that natural assumptions on the local response to those rotations recovers the set of bipartite binary quantum correlations exactly (\cref{sec:txv}), 
 indicating a fundamental relation between the structures of spacetime and of quantum mechanics. 
Finally, we discuss the implications of these results (\cref{sec:Outlook}), 
 particularly for the construction of novel experimental tests of quantum mechanics and of new semi-device-independent protocols for quantum technologies.

\section{Results}
\label{sec:Results}
\subsection{Representation theorem for spatiotemporal degrees of freedom}
\label{sec:BlackBox}
The device-independent formalism abstracts experiments into a table of output statistics conditional on some choice of input.
This is imbued with {\em causal structure}~\cite{Pearl00} 
 by separating the inputs and outputs into local choices and responses made and observed by different {\em local agents}, acting in potentially different locations and times. 
The simplest structure is one agent at a single point in time.
More commonly considered is the {\em Bell scenario}~\cite{BrunnerCPSW14}, 
 where two spatially separated agents each independently select an input (measurement choice) and record the resulting local output.
Theorem~\ref{thm:Rep} of this paper applies to {\em any} casual structure,
 but looking towards application the later examples will use the Bell scenario.

Here, we shall consider experiments where the local inputs correspond to {\em spatiotemporal} degrees of freedom:
 for example, the direction of inhomogeneity of the magnetic field in a Stern--Gerlach experiment, 
 or the angle of a polarization filter (\cref{fig:BellExpt}b). 
Crucially, we will describe such experiments without assuming the validity of quantum mechanics.
 
Let us first consider a single laboratory, say, Alice's. 
For concreteness, assume for the moment that Alice's input is given by the direction $\vec{x}$ of a magnetic field. 
She chooses her input by applying a rotation $R\in\SO{3}$ to some initial magnetic field direction $\vec{x}_0$, i.e.\ $\vec{x}=R\vec{ x}_0$. 
Her statistics of obtaining any outcome $a$ will now depend on this direction, giving her a {\em black box} $\cPr{a}{\vec{x}}$. 

In general, Alice will have a set of inputs $\mathcal{X}$ and a symmetry group $\mathcal{G}$ that acts on $\mathcal{X}$. 
Given some arbitrary $x_0\in\mathcal{X}$, we assume that Alice can generate every possible input $x\in\mathcal{X}$ by applying a suitable transformation $R\in\mathcal{G}$, such that $x=R x_0$. 
Mathematically, $\mathcal{X}$ is then a \emph{homogeneous space}~\cite{Nakahara03}, which can be written $\mathcal{X}=\mathcal{G}/\mathcal{H}$, where $\mathcal{H}\subseteq\mathcal{G}$ is the subgroup of transformations $R'$ with $R' x_0=x_0$.
In the example above, $\mathcal{G}=\SO{3}$ describes the full set of rotations that Alice can apply to $\vec{x}_0$, 
 while $\mathcal{H}=\SO{2}$ describes the subset of rotations that leave $\vec{x}_0$ invariant (i.e.\ the axial symmetry of the magnetic field vector). 
Then, $\mathcal{X}=\SO{3}/\SO{2}=S^2$ is the $2$-sphere of unit vectors (i.e.\ directions) in 3-dimensional space.
Similarly, the polarizer (\cref{fig:BellExpt}b) corresponds to $\mathcal{G}=\SO{2}$, $\mathcal{H}=\{\id\}$, and $\mathcal{X}=S^1$, which we identify with the unit circle. 

Temporal symmetries also fit into this formalism. 
Suppose Alice's input corresponds to letting her system evolve for some time,
 then $\mathcal{G}=\left(\reals,+\right)$ is the group of \emph{time translations}.
If we know that the system evolves periodically over intervals $\tau\in\reals^+$, which we model as a symmetry subgroup $\mathcal{H}=(\tau\cdot\ints,+)$, then the input domain $\mathcal{X}=\mathcal{G}/\mathcal{H}\simeq S^1$.
Physically, this could correspond to applying a controlled-duration Rabi pulse to an atomic system of trusted periodicity before recording an outcome.

Now suppose Alice has a black box $\mathrm{P}$, where on spatiotemporal input $x\in\mathcal{X}$, the outcome $a$ is observed with probability $\cPr{a}{x}$. 
Then, Alice can ``rotate'' her apparatus by $R\in\mathcal{G}$, 
 and induce a new black box $\mathrm{P}'$  with outcome probabilities $\mathrm{P}'\!\left(a\middle|x\right)=\cPr{a}{\!Rx}$.
Physically, $R$ could be an active rotation within Alice's laboratory (e.g.\ spinning a polarizer), of the incident system (e.g.\ adding a phase plate), or could be a passive change of coordinates.

Thus, a given black box and a spatiotemporal degree of freedom defines a family of black boxes, 
 and transformations $R\in \mathcal{G}$ map a given black box to another one in this family.
Suppose we denote the action of $R$ on the black boxes by $T_R: \mathrm{P} \mapsto \mathrm{P}'$.
If rotating the input first by $R$ then by $R'$ is equivalent to a single rotation $R'' = R' \circ R$,
 it follows the black box formed by applying $T_R$ and then $T_{R'}$ is equivalent to applying the single transformation $T_{R''} = T_{R'} \circ T_{R}$ on $\mathrm{P}$.
We can say more about this action if we consider ensembles of black boxes. 
For any family of black boxes $\{\mathrm{P}_i\}_{i=1}^n$ and probabilities $\{\lambda_i\}_{i=1}^n$, $\sum_i \lambda_i=1$, $\lambda_i\geq 0$, the experiment of first drawing $i$ with probability $\lambda_i$ and then applying black box $\mathrm{P}_i$ defines a new, effective black box $\mathrm{P}$, with statistics $\cPr{a}{x}=\sum_i \lambda_i \cPrs{i}{a}{x}$. 
All these black boxes are in principle operationally accessible to Alice. 
However, a priori, we cannot say much about the resulting set of boxes -- it could be a complicated uncountably-infinite-dimensional set defying simple analysis. 
Thus, we make a minimal assumption that this set is not ``too large'':

\begin{assumption}
\label{ass:Finite}
Ensembles of black boxes can be characterized by a finite number of parameters.
\end{assumption}

The mathematical consequence is that the space of possible boxes for Alice is finite-dimensional. 
This is a weaker abstraction of a stronger assumption typically made in
 the \emph{semi-device-independent} framework of quantum information:
 that the systems involved in the protocols are described by Hilbert spaces of bounded (usually small) dimension~\cite{BrunnerPAGMS08,PawlowskiB11}.
For example, BB84~\cite{BennettB84} quantum cryptography assumes that the information carriers are two-dimensional, 
 excluding additional degrees of freedom that could serve as a side channel for eavesdroppers~\cite{AcinGM06}.
Assumption~\ref{ass:Finite} is much weaker; it does not presume that we have Hilbert spaces in the first place.
It is for this assumption (and not the spatiotemporal structure of the input space) that the results presented in this article lie in the {\em semi}-device-independent regime.

We thus arrive at our first theorem. 
Recall that Alice chooses her input $x_R\in\mathcal{X}$ by selecting some $R \in\mathcal{G}$ and applying it to a default input $x_0$, i.e.\ $x=R x_0$.
Then:
\begin{theorem}
\label{thm:Rep}
There is a representation of the symmetry group $\mathcal{G}$ in terms of real orthogonal matrices $R\mapsto T_R$, such that for each outcome $a$, the outcome probabilities $\cPr{a}{x_R}$ are a fixed (over $R$) linear combination of matrix entries of $T_R$.
\end{theorem}
\noindent
The proof is given in \cref{app:GroupRep}, and is based on the observation that $T_R$ becomes a linear group representation on the space of ensembles. Motivated by this characteristic response, we refer to black boxes whose inputs are selected through the action of $\mathcal{G}$ as {\em $\mathcal{G}$--boxes.}

A few comments are in order. 
First, this theorem applies to any causal structure, including the case of two parties performing a Bell experiment. 
If Alice and Bob have inputs and transformations $\mathcal{X}_A,$ $\mathcal{G}_A$ and $\mathcal{X}_B$, $\mathcal{G}_B$ respectively, then the full setup can be seen as an experiment with $\mathcal{X}=\mathcal{X}_A\times\mathcal{X}_B$ and $\mathcal{G}=\mathcal{G}_A\times\mathcal{G}_B$, to which Theorem~\ref{thm:Rep} applies directly.

Secondly, there may be more than one transformation that generates the desired input $x$, 
 i.e.\ both $x=Rx_0$ and $x=R' x_0$ for $R\neq R'$; this is precisely the case if $R^{-1}R'\in\mathcal{H}$. 
For example, a magnetic field can be rotated from the $y$- to $z$-direction in many different ways. 
In this case, Theorem~\ref{thm:Rep} applies to both $R$ and $R'$, which yields additional constraints.

Finally, quantum theory is contained as a special case. 
Typically, one argues that due to preservation of probability, transformations $R$ must be represented in quantum mechanics via unitary matrices $U_R$ acting on density matrices via $\rho \mapsto U_R \rho U_R^\dagger$. 
This projective action can be written as an orthogonal matrix on the real space of Hermitian operators, in concordance with Theorem~\ref{thm:Rep}. 

As a specific example, consider a quantum harmonic oscillator with frequency $\omega$,
 initially in state $\rho_0$,
 left to evolve for a variable time $t$ before it is measured by a fixed POVM~\cite{NielsenC00} $\{M_a\}_{a\in\mathcal{A}}$.
The free dynamics are given by the Hamiltonian $H$,
 whose discrete set of eigenvalues $\{E_n = \hbar\omega\left(\frac{1}{2}\!+\!n\right)\}$ correspond to allowed ``energy levels''.
The evolution is periodic, so (recalling earlier) $\mathcal{G} = \left(\reals,+\right)$, $\mathcal{H} = \left(\frac{2\pi}{\omega}\cdot\ints,+\right)$ and $\mathcal{X}\simeq\sphere{1}$.
The associated black box is thus 
 $\cPr{a}{t} = \tr\left[M_a \exp\left(-\frac{iHt}{\hbar}\right) \rho_0 \exp\!\left(\frac{iHt}{\hbar}\right)\right]$. 
For any given $\rho_0$ and $M_a$, this evaluates to
 an affine-linear combination of terms of the form $\cos\left[\left(n-m\right)\hbar\omega t\right]$ and $\sin\left[\left(n-m\right)\hbar\omega t\right]$, 
 involving all pairs of energy levels that have non-zero occupation probability in $\rho_0$ (and non-zero support in $M_a$). 
This is a linear combination of entries of the matrix representation
\begin{equation}
\label{eq:timeTrans}
   T_t=\bigoplus_{\alpha=E_n-E_m}\left(\begin{array}{cc} \cos(\alpha t) & \sin(\alpha t) \\ -\sin(\alpha t) & \cos(\alpha t)\end{array}\right),
\end{equation}
in accordance with Theorem~\ref{thm:Rep}. 
For $T_t$ to be a finite matrix, there must only be a finite number of occupied energy differences $E_m-E_n$. 

Here, Assumption~\ref{ass:Finite} is equivalent to an upper (and lower) bound on the system's energy.
In the general framework that does not assume the validity of quantum mechanics (or presuppose trust in our devices, or our assignment of Hamiltonians), 
 we can view Assumption~\ref{ass:Finite} as a natural generalization of this to other symmetry groups and beyond quantum theory. 
By assuming a concrete upper bound on the representation label (such as $\alpha$ in \cref{eq:timeTrans}), we can establish powerful theory- and device-independent consequences for the resulting correlations, as we will now demonstrate by means of several examples.

\subsection{Example: Two angles and Bell witnesses}
\label{sec:SO2}
Let us consider the simplest non-trivial spatiotemporal freedom,
 where Alice and Bob each have the choice of a single continuous angle: respectively $\alpha, \beta \in [0,2\pi)$,
 and each obtain a binary output $a,b\in\{+1,-1\}$.
Physically, this would arise, say, 
 in experiments where a pair of photons is distributed to the two laboratories,
 each of which contains an rotatable polarizer followed by a photodetector (\cref{fig:BellExpt}b).

Due to Theorem~\ref{thm:Rep}, the probabilities $\cPr{a,b}{\alpha,\beta}$ are linear combinations of matrix entries of an orthogonal representation of $\SO{2}\!\times\!\SO{2}$. 
From the classification of these representations (see \cref{app:SO2xSO2}), it follows that all $\SO{2}\!\times\!\SO{2}$-boxes are of the form
\begingroup
\addtolength{\jot}{-1em}
\begin{align}
\label{eq:GenericProb}
\cPr{a,b}{\alpha,\beta} := &\sum_{m=0}^{2J} \sum_{n=-2J}^{2J} c^{ab}_{mn} \cos\left( m \alpha - n \beta\right) \nonumber \\
&\hspace{6em} + s^{ab}_{mn} \sin\left( m \alpha - n \beta\right),
\end{align}
\endgroup
resulting in a correlation function 
\begingroup
\addtolength{\jot}{-0.5em}
\begin{align}
\label{eq:CorrFunc}
C\!\left(\alpha, \beta\right) & :=  \cPr{+1,+1}{\alpha, \beta} + \cPr{-1,-1}{\alpha, \beta}  \nonumber \\[0.5em]
& \qquad  - \cPr{+1,-1}{\alpha,\beta} - \cPr{-1,+1}{\alpha,\beta}  \\ 
& \hspace{-3.5em} = \sum_{m=0}^{2J} \sum_{n=-2J}^{2J} \hspace{-0.25em} C_{mn} \cos\left( m \alpha\!-\!n \beta\right) + S_{mn} \sin\left( m \alpha\!-\!n \beta\right), 
\nonumber\\[-0.75em]
\label{eq:GenericCorr}
\end{align}   
\endgroup
where $J\in\{0,\frac 1 2,1,\frac 3 2,\ldots\}$ is some finite maximum ``spin''.

If Alice and Bob's laboratories are spatially separated, the laws of relativity forbid Alice from sending signals to Bob instantaneously. 
This ``no-signalling'' principle constrains the set of valid joint probability distributions:
 namely Bob's marginal statistics cannot depend on Alice's choice of measurement, and vice versa. 
However, for {\em any} given correlation function of the form \cref{eq:GenericCorr}, there is always at least one set of valid no-signalling probabilities (see \cref{app:AnyCIsNoSig}) --
for example, those where the marginal distributions are ``maximally mixed'' such that independent of $\alpha$, $a$ is $+1$ or $-1$ with equal probability (likewise for $\beta$ and $b$),
 consistent with an observation of \citet{PopescuR94}.

Consider a quantum example: two photons in a Werner state~\citep{Werner89,AugusiakDA14} 
  $\rho_W := p \ketbra{\psi^-}{\psi^-}  + \frac{1}{4} (1\!-\!p) \id_4$
 where $\ket{\psi^-} = \frac{1}{\sqrt{2}} \left( \ket{0}\!\ket{1} - \ket{1}\!\ket{0}\right)$ and $p\in[0,1]$.
Alice and Bob's polarizer/detector setups are described by the observables
 $M_\theta := \left(\begin{smallmatrix} \cos 2\theta & \sin2\theta  \\ \sin2\theta & -\cos2\theta \end{smallmatrix}\right)$
 for orientations $\theta=\alpha,\beta$ respectively.
Then, $\corr{\alpha,\beta} = \tr \left(\rho_W M_\alpha\!\otimes\!M_\beta\right) = -p \cos \left[2\left(\alpha - \beta\right)\right]$.
This fits the form of \cref{eq:GenericCorr} for $J=1$, with $C_{22}=-p$ and all other coefficients as zero.

A paradigmatic question in this setup is whether the statistics can be explained by a {\em local hidden variable} (LHV) model.
Namely,
is there a single random variable  $\lambda$ over some space $\Lambda$
such that $\cPr{a,b}{\alpha,\beta} = \int_{\Lambda} \mathrm d\lambda \, \Prs{\Lambda}{\lambda} \cPrs{A}{a}{\alpha,\lambda} \cPrs{B}{b}{\beta,\lambda}$,
 where $\Prs{\Lambda}{\lambda}$ is a classical probability distribution,
 and $\cPrs{A}{a}{\alpha,\lambda}$ and $\cPrs{B}{b}{\beta,\lambda}$ are respectively Alice and Bob's local response functions (conditioned on their input choices $\alpha$ and $\beta$ and the particular realization of the hidden variable $\lambda$)?
If no LHV model exists, then the scenario is said to be \emph{nonlocal}. 
Famously, Bell's theorem shows that quantum theory admits correlations that are nonlocal in this sense~\cite{EinsteinPR35,Bell64}. 
This follows from the violation of \emph{Bell inequalities} that are satisfied by all distributions with LHV models,
 the archetypical example being the Clauser--Horne--Shimony--Holt (CHSH) inequality~\cite{ClauserHSH69}:
\begin{equation}
\label{eq:CHSH}
\big| \corr{\alpha_1,\beta_2} + \corr{\alpha_3,\beta_2} + \corr{\alpha_3,\beta_4} - \corr{\alpha_1,\beta_4}\big| \leq 2,
\end{equation}
where $\alpha_1,$ $\alpha_3$ are two choices of Alice's angle, and $\beta_2$, $\beta_4$ of Bob's.
Classical systems always satisfy this bound,
 but quantum theory admits states and measurements that violate it.
When working with a continuous parameter, Bell inequalities need not be limited to a subset of angles,
 but can also be formulated as a {\em functional} of the entire correlation function~\cite{Zukowski93,SenDeSZ02}.
 
Not all correlations of the form in~\cref{eq:GenericCorr} are allowed by quantum theory. 
For example, ``science fiction'' polarizers with the correlation function $C(\alpha,\beta)=\frac{2}{7}\cos[3(\alpha-\beta)]-\cos[\alpha-\beta]$ would yield a CHSH value of 3.63, under choices of angles $\alpha_1=1.5$, $\alpha_3=0$, $\beta_2=3.9$ and $\beta_4=2.3$, violating quantum theory's maximum achievable value of $2\sqrt{2}$~\cite{Cirelson80}.

With this general form, we can make broad statements about whether correlations are local or nonlocal.
First, if the correlations are sufficiently ``noisy'', we can systematically construct a LHV model by generalizing a procedure by \citet{Werner89} (see \cref{app:LHV}).
If the only constraint on the correlations is that is has some maximum $J$, then
 the existence of a LHV is guaranteed if the magnitude of angle-dependent changes in $C$ is less than $\gamma_J$ where
\begin{equation} 
\label{eq:WeakBound}
\gamma_J := \sqrt{2} e^{-1} \left[4J\left(2J+1\right)\right]^{-\frac{3}{2}}.
\end{equation}
Subject to extra restrictions that keep the form of $C$ simple, more permissive bounds are also derived.
For instance, if there is only one non-zero coefficient in \cref{eq:GenericCorr}, then $\gamma = \sqrt{2}/\pi \approx 0.4502$. Recall the correlation function for projective measurements on a Werner state, $-p \cos\left[2\left(\alpha-\beta\right)\right]$,
 and identify $\gamma$ with $p$.
In this case, our bound is comparable with that in \citet{HirschQVNB17} of $p \leq 0.6829$.

Conversely, we can give a simple sufficient criterion for nonlocality if we separate the terms in \cref{eq:GenericProb,eq:GenericCorr} into {\em relational} and {\em non-relational} components.
The relational components where $m=n$ account for behaviour that depends only on the difference between the two angles.
Purely relational correlations, i.e.\ ones with $\corr{\alpha,\beta}\!\equiv\!\corr{\alpha\!-\!\beta}$, can be motivated by symmetry (i.e.\ that in the absence of external references, {\em only} the relative angle should have operational meaning).
Here, the $J=\frac{1}{2}$ case contains the bipartite {\em rotational invariant} correlations discussed in \citet{NagataLWZ04}.
Conversely, the correlations resulting from any experiment can be actively made relational as we will describe in more detail below.

If the relational part of a correlation function $C_{\rm rel}$ has an angle difference $\Theta_+$ which results in near perfect \mbox{(anti-)}correlations,
 and another angle difference $\Theta_-$ that does not,
 then one can systematically construct a (Braunstein--Caves~\cite{BraunsteinC90}) Bell inequality that will be violated (see \cref{app:NonLocal}).
Specifically, ``near perfect'' means that for a given $J$, $C_{\rm rel}\left(\Theta_+\right) \geq 1 - \varepsilon_J$ with 
\begin{align}
\label{eq:StrongBound}
\varepsilon_J := - K_J+\sqrt{K_J^2+\frac{\Delta^2}4} = \frac{\Delta^2}{8 K_J}+\order{K_J^{-2}},
\end{align}
where $K_J:=\sqrt{2} \pi^2 J(2J+1)(4J+1)/3$, and $C_{\rm rel}\!\left(\Theta_-\right) \leq 1 - \Delta$ bounds the ``other'' value measured at $\Theta_-$.
(See \cref{app:NonLocal} for proof).

We summarize these results (see also \cref{fig:Curves}):
\begin{theorem}
\label{thm:SO2xSO2}
Consider a two-angle Bell experiment with correlations $C$ in the form of \cref{eq:GenericCorr}, 
 with an upper bound $J$ on the representation labels.
\begin{enumerate}[A.]
\item If $C$ is sufficiently ``noisy'', in the sense that
\begin{equation}
   \max_{\alpha,\beta}|C(\alpha,\beta)-C_{00}|\leq \gamma_J (1-|C_{00}|)
\end{equation}
 with $\gamma_J$ as in \cref{eq:WeakBound},
 then the correlations can always be exactly accounted for by a LHV model.
\item If the relational part of $C$ is sufficiently ``pure'' for some angle $\Theta_+$ (above $1-\varepsilon_J$, as defined in \cref{eq:StrongBound}), but also sufficiently different (below $1-\Delta$) for some other angle $\Theta_-$, then the correlations violate a Bell inequality.
\end{enumerate}
\end{theorem}

\begin{figure}[bht]
\includegraphics[width=0.42\textwidth]{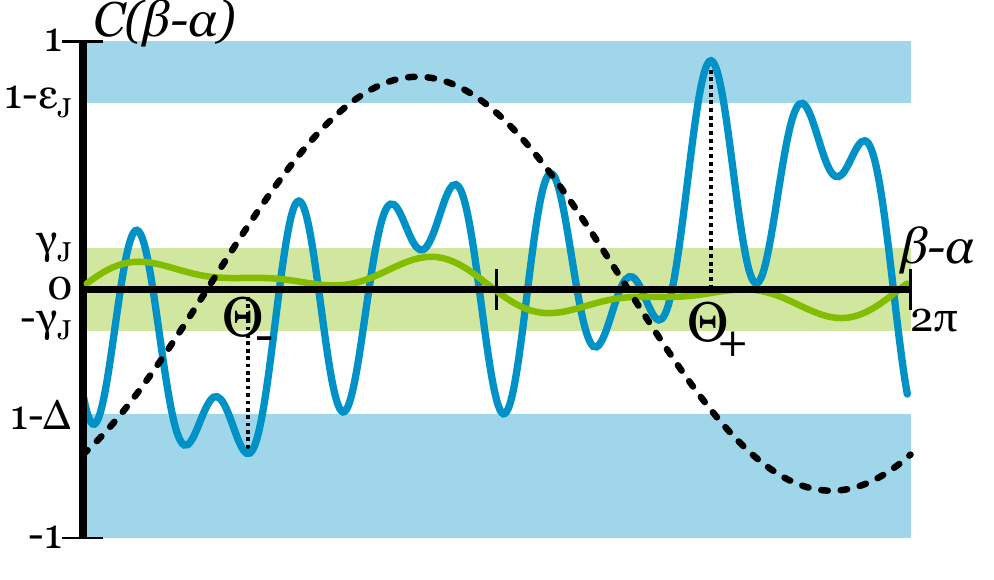}
\caption{
\label{fig:Curves}
\caphead{Two-angle relational correlation functions.} \\
A ``sufficiently noisy'' correlation function can always be reproduced exactly by a LHV model (Theorem~\ref{thm:SO2xSO2}a).
This is represented by the green curve completed contained within the central green-shaded region (drawn for $C_{00}=0$).
Conversely, if the function is ``pure enough'', then it must be nonlocal (Theorem~\ref{thm:SO2xSO2}b). 
This is represented by the blue curve with values in both extremal blue-shaded regions.
Not all curves can be realized within quantum theory, but simple sinusoidal curves certainly can (such as the dashed black curve), following from Theorem~\ref{thm:222quantum} in two dimensions.}
\end{figure}

This is a powerful result:
 with a choice between two experimental settings for Alice, and no choice made by Bob,
 we can {\em witness} nonlocality. This can be done by the following protocol:
\begin{itemize}
\setlength{\itemsep}{0em}
 \item Alice and Bob share some random angle $\lambda$, uniformly distributed in the interval $[0,2\pi)$.
 \item Alice chooses locally freely between the two possible angles $\alpha\in\{\Theta_+,\Theta_-\}$.
 \item Alice now inputs $\alpha+\lambda$ into her half of the box, while Bob inputs $\lambda$.
 \item By repeating the protocol, they determine the correlations $C_{\rm rel}(\Theta_+)$ and $C_{\rm rel}(\Theta_-)$, and verify that they violate the inequality above. 
\end{itemize}
Randomization over $\lambda$ effectively projects $C$ onto its relational part 
 $C_{\rm rel}(\alpha,\beta)=\frac 1 {2\pi} \int_0^{2\pi} \mathrm{d}\lambda \, \corr{\alpha\!+\!\lambda,\beta\!+\!\lambda}$,
 which only depends on $\alpha-\beta$. 
The protocol above fixes $\beta$ to zero, while $\alpha\in\{\Theta_+,\Theta_-\}$. This is sufficient to determine the two correlation values.

The protocol assumes that Alice and Bob have some physically motivated promise on the maximum representation label $J$ (e.g.\ by assuming an upper bound on the total energy of the system, or the number of elementary particles transmitted), and that they know the angles $\Theta_+$ and $\Theta_-$ beforehand. 
The latter assumption is analogous to standard Bell experiments, where the relevant measurement settings are assumed to be known.

Witnessing Bell nonlocality is not the same as directly demonstrating nonlocality (i.e.\ collecting all the statistics for a Bell test, which is only possible if Bob has some free choice too) but rather, subject to Assumption~\ref{ass:Finite}, implies the existence of an experiment that \emph{would} demonstrate nonlocality. 
In contrast to a full Bell experiment, a Bell witness has the advantage of being experimentally easier to implement: the protocol above allows one to witness nonlocality with only \emph{two} measurement settings instead of four. 
Note that \emph{only} making the correlation function relational (i.e.\ going from $C$ to $C_{\rm rel}$ as above) \emph{without} any additional assumption on $J$ is \emph{not} sufficient to obtain this reduction, as we show in~\cref{app:Necessity}.

Our protocol hence demonstrates that natural assumptions on the response of devices to spatiotemporal transformations can give additional constraints that allow for the construction of new Bell witnesses.
 This opens up the possibility of new methods of experimentally certifying nonlocal behaviour, similar to~\cite{TuraASVLA14,SchmiedBAFSTS16,WangSN17}, but without the need to presume the validity of quantum theory or to trust all involved measurement devices.

Theorem~\ref{thm:SO2xSO2} shows us that smaller values of $J$ (and hence ``simpler'' responses to changes in angles) result in more permissive bounds for finding LHV models, or witnessing non-locality. 
In our next example, we shall move from angles ($\SO{2}$) to directions ($\SO{d}$), but consider arguably the simplest non-trivial response.

\subsection{Example: Characterizing quantum correlations}
\label{sec:txv}
For our last example, 
 we shall apply our framework to characterize the set of correlations that can be realized by two parties sharing a quantum state, each locally choosing one of two binary-outcome measurements -- the thus called quantum ``(2,2,2)''-behaviours.
The set of quantum $(2,2,2)$-behaviours $\mathcal{Q}$ is a strict superset of the classical $(2,2,2)$-behaviours $\mathcal{C}$ (i.e.\ those admitting a LHV model). 
However, the set of all no-signalling behaviours $\mathcal{NS}$ is strictly larger: $\mathcal{C}\subsetneq \mathcal{Q}\subsetneq \mathcal{NS}$~\cite{KhalfinT85,PopescuR94}. 
This has led to the search for simple physical or information-theoretic principles that would explain ``why'' nature admits no more correlations than in $\mathcal{Q}$.
Several candidates have been suggested over the years, including {\em information causality}~\cite{PawlowskiPKSWZ09}, {\em macroscopic locality}~\cite{YangNSS11}, or {\em non-trivial communication complexity}~\cite{BrassardBLMTU06}, but none of these have been able to single out $\mathcal{Q}$ uniquely~\cite{NavascuesGHA15}.

Here, we will provide such a characterization by considering black boxes that transform in arguably the simplest manner.
Over a spherical input domain $\mathcal{X}=\sphere{d-1}$ 
 an $\SO{d}$--box $\cPr{a}{\vec{x}}$
 is said to {\em transform fundamentally} if the representation matrix $T_R$ in Theorem~\ref{thm:Rep} 
 can be chosen as the block matrix $\id_1\oplus R$, 
 where $\id_1:=\left(1\right)$ and $R$ is the fundamental representation of $\SO{d}$ (e.g.\ for $d=3$, $\{R\}$ are the familiar rotation matrices).
Consequently, a black box that transforms fundamentally has an {\em affine} representation, $\cPr{a}{\vec{x}} = c_0^a + \vec{c}^{\,a}\cdot\vec{x}$ where $\vec{x}\in\sphere{d-1}$ is the input, and $c_0^a\in\reals^+$, $\vec{c}^{\,a}\in\reals^d$  (proof in \cref{app:Vector}).
 
Motivated by symmetry, we consider a class of {\em unbiased} black boxes that do not prefer any particular output when averaged over all possible inputs. 
This implies that $c_0^a = 1/|\mathcal{A}|$ for every $a$.
For example, this symmetry holds for measurements on quantum spin-$\frac{1}{2}$ particles: spin $+\frac{1}{2}$ in one direction is the same as spin $-\frac{1}{2}$ in the opposite, and hence neither outcome is preferred on average. 

Imagine Alice and Bob residing in $d$-dimensional space ($d\geq 2$), sharing a non-signalling box $\cPr{a,b}{\vec{x},\vec{y}}$, where both inputs $\vec x,\,\vec y\in S^{d-1}$ are spatial directions, and $a,\,b$ each can take two values. 
Suppose that their {\em conditional boxes} transform fundamentally and are unbiased.
A conditional box ${\rm P}_A^{b,\vec{y}}\left(a\middle|\vec{x}\right) := \cPr{a,b}{\vec{x},\vec{y}} / \cPrs{B}{b}{\vec{y}}$ describes the local black box Alice would have if she was told Bob's measurement choice $\vec{y}$ and outcome $b$.
If all conditional boxes for Alice and Bob transform fundamentally, then the bipartite box is said to {\em transform fundamentally locally}.
Similarly, if all conditional boxes are unbiased, $\cPr{a,b}{\vec{x},\vec{y}}$ is said to be {\em locally unbiased}.

Surprisingly, these local symmetries severely constrain the \emph{global} correlations: 
 they allow for only and exactly those correlations that can be realized by two parties who share a \emph{quantum} state and choose between two possible two-outcome quantum measurements each---the \emph{quantum $(2,2,2)$-behaviours}:
\begin{theorem}
\label{thm:222quantum}
The quantum $(2,2,2)$-behaviours are \mbox{exactly} those that can be realised by binary-outcome \mbox{bipartite} {$\SO{d}\!\times\!\SO{d}$-}boxes that {\em transform fundamentally locally} and are {\em locally unbiased}, restricted to two choices of input direction per party per box, and statistically mixed via shared randomness.
\end{theorem}
\noindent
The proof is given in~\cref{app:Vector}.

A few remarks are in place. 
First, the \emph{unbiasedness} refers to the total set $\sphere{d-1}$ of possible inputs per party, not to the two inputs to which the box is restricted. 
Even if the unrestricted behaviour is unbiased in the sense described above, the resulting $(2,2,2)$-behaviour can be biased.
Secondly, this unbiasedness of the underlying {$\SO{d}\!\times\!\SO{d}$-}box is necessary to recover the quantum correlations -- without it, one can realize arbitrary nonsignalling correlations, including PR--box behaviour, in a way that still transforms locally fundamentally (we give an example in \cref{app:Vector}).
Finally, shared randomness is necessary to realize explicitly {\em non-extremal} quantum correlations by such boxes,
 following on the observation that the set of $(2,2,2)$--behaviours realizable by POVMs on two qubits is not convex~\cite{PalV09,DonohueW15}. 
 Namely, if both parties share the $(2,2,2)$-behaviours $P_0$ and $P_1$ and a random bit $c\in\{0,1\}$ that equals $0$ with probability $\lambda$, they can statistically implement the mixed behaviour $\lambda P_0+(1-\lambda)P_1$ by feeding their inputs into box $P_c$.

For $n\!=\!2$ parties with $m\!=\!2$ measurements and $k\!=\!2$ outcomes each, our result provides a characterization of the quantum set.
Although Theorem~\ref{thm:222quantum} cannot be extended to general $(m,n,k)$-behaviours~\cite{AcinACHKLMP10} without modification,
 this result shows that our framework of $\mathcal{G}$-boxes offers a very natural perspective on physical correlations, and reinforces earlier observations that hint at a deep fundamental link between the structures of spacetime and quantum theory~\cite{Wootters80, MuellerM13,HoehnM16,GarnerMD17}.

\section{Discussion and outlook}
\label{sec:Outlook}
We have introduced a general framework for semi-device-independent information processing, without assuming quantum mechanics,
 for black boxes whose inputs are degrees of freedom that break spatiotemporal symmetries. 
Such black boxes have \mbox{characteristic} \mbox{probabilistic} responses to symmetry transformations, 
 and natural assumptions about this behaviour
 can certify technologically important properties like the presence or absence of Bell correlations.

Specifically, we have shown that the quantum $(2,2,2)$-behaviours can be exactly classified as those of bipartite boxes that transform \emph{locally} in the simplest possible way 
 -- by the fundamental representation of $\SO{d}$ rotations, respecting the unbiasedness of outcomes. 
For Bell experiments with $\SO{2}\times\SO{2}$-boxes, 
 we have shown that correlations that are quantifiably ``noisy enough'' always admit a local hidden variable model, 
 whereas relational correlations for which there are settings with differing ``purity'' must violate a Bell inequality. 
Since the underlying technical tools (e.g.\ Schur orthogonality~\cite{Sepanski07}) hold in greater generality, 
 many of our results could be applied to other groups.

Furthermore, these results have allowed us to construct a protocol to \emph{witness} the violation of a Bell inequality within a causal structure that is otherwise too simple to admit the direct detection of nonlocality.
We believe that our approach can be applied to experimental settings, 
 such as the recent demonstration of Bell correlations in a Bose--Einstein condensate~\cite{SchmiedBAFSTS16}, 
 and potentially eliminate the necessity to trust all detectors or to assume the exact validity of quantum mechanics. 
Many of these experiments \emph{do} work with spatiotemporal inputs like Rabi pulses, which makes our approach particularly natural for analyzing them.

We have predominantly worked under the assumption that ensembles of black boxes are characterized by a \emph{finite} number of parameters, and -- more specifically~-- that an upper bound on the representation label (say, the ``spin'' $J$) of the boxes is known. 
On one hand, this assumption can likely be weakened, by employing group-theoretic results such as the Peter--Weyl theorem~\cite{Sepanski07}. 
On the other hand, we have argued that this assumption is natural: 
 it is weaker than assuming a Hilbert space with bounded dimension 
 (standard in the semi-device-independent framework~\cite{PawlowskiB11})
 and constitutes a generalization of an ``energy bound'' beyond quantum theory (cf.~\cite{BranfordDG18}). 
Moreover, it incorporates an intuition conceptually closer to particle physics: to quantify the potential eavesdropping side channels, one might not count Hilbert space dimensions, but rather representation labels, since these are intuitively (and sometimes rigorously) related to the total number of particles.

Our framework opens up several potential avenues for future work. 
First, as the witness example demonstrates, 
 our formalism hints at novel semi-device-independent protocols based on assumptions with firmer physical motivation than the usual dimension bounds. 
In contrast to recent proposals for using energy bounds~\cite{vanHimbeeckWCGP17,VanHimbeeckP19,Rusca19}, our assumption on the devices' symmetry behaviour does not presume the validity of quantum mechanics, but rather embodies a natural upper bound to the ``fine structure'' of the devices' response. 
Meanwhile, one might apply the functional approach~\cite{Zukowski93,SenDeSZ02} to our framework by taking Haar integrals over spatiotemporal input spaces to derive a device--independent family of generalized Bell--\.{Z}ukowski inequalities for various limits of fine structure.

Secondly, our framework informs novel experimental searches for conceivable physics beyond quantum theory. 
Previous proposals (e.g.\ superstrong nonlocality~\cite{PopescuR94} or higher-order interference~\cite{Sorkin94,UdudecBE10}) have simply described the probabilistic effects without predicting how they could actually occur within spacetime as we know it. 
This has made the search for such effects seem like the search for a needle in a haystack~\cite{SinhaCJLW10}. 
Our formalism promises a more direct spatiotemporal description of such effects -- hopefully leading to predictions that are more tied to experiments and in greater compatibility with spacetime physics.

Combining the principles of quantum theory with special relativity has historically been an extremely fruitful strategy. 
Here, we propose to extend this strategy to device-independent quantum information and even beyond quantum physics. 
In principle, suitable extensions of our framework would allow us to address questions such as: \emph{which probability rules are compatible with Lorentz invariance}? 
Any progress on these kind of questions has the potential to give us fascinating insights into the logical architecture of our physical world.

\acknowledgments
We are grateful to Miguel Navascu\'es, Matt Pusey, and Valerio Scarani for discussions.
This project was made possible through the support of a grant from the John Templeton Foundation. 
The opinions expressed in this publication are those of the authors and do not necessarily reflect the views of the John Templeton Foundation.
We acknowledge the support of the Austrian Science Fund (FWF) through the Doctoral Programme CoQuS.
This research was supported in part by Perimeter Institute for Theoretical Physics. 
Research at Perimeter Institute is supported by the Government of Canada through the Department of Innovation, Science and Economic Development Canada and by the Province of Ontario through the Ministry of Research, Innovation and Science. 

%

\balancecolsandclearpage

\newpage
\clearpage
\appendix

\section{The representation of spatiotemporal degrees of freedom in black box statistics}
\label{app:GroupRep}
Let us first furnish a mathematical description of a black box as an input--output process. We begin with the single party case (say, Alice).
Suppose the domain of Alice's inputs is the set $\mathcal{X}$, and of her outputs is the finite set $\mathcal{A}$. 
As motivated in the main text, we are interested in the case where $\mathcal{X}$ is a homogeneous space. 
That is, we have a group $\mathcal{G}$ that acts transitively on the set of inputs $\mathcal{X}$, such that $\mathcal{X}=\mathcal{G}/\mathcal{H}$, and $\mathcal{H}\subseteq\mathcal{G}$ is the corresponding stabilizer subgroup. 
The paradigmatic example is given by $\mathcal{X}=S^{d-1}$, $\mathcal{G}=\SO{d}$ and $\mathcal{H}=\SO{d\!-\!1}\subset \mathcal{G}$, such that the inputs $\vec x \in\mathcal{X}$ are unit vectors. 
Even though the inputs need not be vectors in general, we will use the vector notation in the following for convenience. 
We will assume that $\mathcal{G}$ is a locally compact group, such that all bounded finite-dimensional representations are unitary~\cite{FiniteDim}.

For such an input domain, we can assign an arbitrary ``default input'' $\vec x_0\in\mathcal{X}$, such that every other input $\vec x\in\mathcal{X}$ can be written as $\vec x=R_{\vec x}\vec x_0$ for some suitable transformation $R_{\vec x}\in\mathcal{G}$. 
Physically, we can imagine that Alice chooses her input by ``rotating'' the default input $\vec x_0$ into her desired direction $\vec x$, and she can do so by applying a suitable rotation $R_{\vec x}$. 
In general, $R_{\vec x}$ is not unique, and Alice's freedom of choice of transformation is given by $\mathcal{H}$.

A {\bf black box} $\mathrm{P}$ is then a map $\mathrm{P}: \mathcal{X} \to \reals^{|\mathcal{A}|}$ such that for $\vec{x}\in\mathcal{X}$,
$\mathrm{P}^a: \vec{x} \mapsto \cPr{A=a}{X=\vec{x}}$, where $\mathrm{P}^a$ is the $a^{\rm th}$ element of the vector map.
Since for probabilities $0 \leq\cPr{A=a}{X=\vec{x}}\leq 1$,
 each $\mathrm{P}^a$ is a non-negative real bounded function on $\mathcal{X}$.
For probabilities, we also have the constraint that for all $\vec{x}$, $\sum_a \mathrm{P}^a \vec{x} = 1$; so the range of the vector function $\mathrm{P}$ is actually that of $\left(|\mathcal{A}|-1\right)$--dimensional simplices (a compact convex subspace of $\reals^{|\mathcal{A}|}$).
As such, $\mathrm{P}\in\mathcal{B}(\mathcal{X})^{|\mathcal{A}|}$ where $\mathcal{B}(\mathcal{X})$ is the set of bounded functions on $\mathcal{X}$.

\begin{definition}[$\mathcal{G}$-box]
A black box (formalized above) whose input domain $\mathcal{X}$ is a homogeneous space acted transitively upon by the group $\mathcal{G}$ is known as a $\mathcal{G}$-box.
\end{definition}

\vspace{0.5em}
\noindent {\bf Proof of Theorem 1.}
{\em
Consider a $\mathcal{G}$-box whose ensemble behaviour can be characterized by a finite number of parameters (Assumption~\ref{ass:Finite}).
There is a representation of the symmetry group $\mathcal{G}$ in terms of real orthogonal matrices $R\mapsto T_R$, such that for each outcome $a$, the outcome probabilities $\cPr{a}{x_R}$ are a fixed (over $R$) linear combination of matrix entries of $T_R$.
}
\begin{proof} 

Suppose Alice has a black box $\mathrm{P}$,
 and access to a geometric freedom $\mathcal{G}$ acting on $\mathcal{X}$.
For each $R\in\mathcal{G}$, Alice can induce a new black box $\mathrm{P}'$ by first applying $R$ to her input $\vec{x}$ and then supplying the input $R\vec{x}$ to $\mathrm{P}$,
 which acts as $\mathrm{P}'^a: \vec{x} \mapsto \cPr{a}{R\vec{x}}$, i.e.\ $\mathrm{P}'\!\left(a\middle|\vec{x}\right)=\cPr{a}{\!R\vec{x}}$.
 
For each $R$, we can define a map $T_R: \mathrm{P} \mapsto \mathrm{P}'$, acting on each component of $\mathrm{P}$ via $T_R \mathrm{P}^a=\mathrm{P}'^a$. 
Obviously, $T_R\circ T_S=T_{RS}$, so if we denote the ``space of black boxes'' accessible to Alice by $\Omega_{\mathcal{G}} := \{T_R \mathrm{P} \,|\, R\in\mathcal{G}\} \subseteq \mathcal{B}(\mathcal{X})^{|A|}$, then $T_R$ defines a group action on $\Omega_{\mathcal{G}}$.

Consider the linear extension $\Omega_{\mathcal{G}}^\reals := \mathrm{span}\left(\Omega_\mathcal{G}\right)$, a linear subspace of $\mathcal{B}(\mathcal{X})^{|A|}$, with elements $Q = \sum_{i=1}^n \lambda_i \mathrm{P}_i$, where $n\in\nats$ is arbitrary but finite, all $\lambda_i \in \reals$, and $\mathrm{P}_i\in \Omega_\mathcal{G}$.
Note $Q: \mathcal{X} \to \reals^{|\mathcal{A}|}$, but without further restriction on $\{\lambda_i\}$ this may map to outside of the simplex of normalized probabilities.

Now, consider the effect of $R\in\mathcal{G}$ on some object $Q$. 
Since $Q: \vec{x}\mapsto \sum_i \lambda_i \cPr{a}{\vec{x}}$,
 applying $R$ first to take $R:\vec{x}\mapsto\vec{x}'$ gives us $Q\circ R: \vec{x} \mapsto \sum_i \lambda_i \cPr{a}{R\vec{x}}$,
  and hence $Q\circ R = \sum_i \lambda_i T_R P_i$. 
Since $T_R P=P\circ R$ for $P\in\Omega_{\mathcal{G}}$, we can define the map $\tilde T_R:\Omega_{\mathcal{G}}^{\reals}\to \Omega_{\mathcal{G}}^{\reals}$ via $\tilde T_R Q:=Q\circ R$ as an extension of the map $T_R$. By construction, every $\tilde T_R$ is a linear map, and
\begin{equation}
     \tilde T_R \tilde T_S(Q)=Q\circ R\circ S=Q\circ(R\circ S)=\tilde T_{RS}(Q),
\end{equation}
 hence $R\mapsto \tilde T_R$ is a real linear representation of $\mathcal{G}$. Since $\tilde T_R$ is an extension of $T_R$, we drop the tilde from our notation. 
As we have assumed that ensembles of black boxes can be characterized by a finite number of parameters, 
 the linear space $\Omega^\reals_\mathcal{G}$ is finite-dimensional.
Then $T_R$, as linear maps acting on a finite-dimensional real vector space, may be expressed as real matrices.

Next, we need to show that the representation $R\mapsto T_R$ is \emph{bounded}, i.e.\ that $\sup_{R\in\mathcal{G}} \|T_R\|<\infty$.
This will exclude, for example, cases like $\mathcal{G}=\left(\reals, +\right)$ and $T_t:=\left(\begin{smallmatrix} 1 & t \\ 0 & 1  \end{smallmatrix}\right)$. 
To this end, let $P_1,\ldots,P_D\in \Omega_{\mathcal{G}}$ be a linearly independent set of boxes that spans $\Omega_{\mathcal{G}}^{\reals}$ (that is, a basis of boxes, hence $D=\dim \Omega_{\mathcal{G}}^{\reals}$).
Then, every ${\rm P}\in \Omega_{\mathcal{G}}^{\reals}$ has a unique representation ${\rm P}=\sum_{i=1}^D \alpha_i P_i$, and $\|{\rm P}\|_1:=\sum_{i=1}^D |\alpha_i|$ defines a norm on $\Omega_{\mathcal{G}}^{\reals}$. 
We can define another norm on this space via
\begin{equation}
   \left\|P\right\|:=\sup_{\vec x\in\mathcal{X}} \sum_{a\in\mathcal{A}}  \left|\cPr{a}{\vec{x}}\right|.
\end{equation}
This is finite since ${\rm P}\in\mathcal{B}(\mathcal{X})^{|\mathcal{A}|}$, and it is easy to check that it satisfies the properties of a norm.
Since all norms on a finite-dimensional vector space are equivalent, there is some $c>0$ such that $\|\bullet\|_1\leq c\|\bullet\|$. Furthermore, all ${\rm P}'\in\Omega_{\mathcal{G}}$ satisfy $\|{\rm P}'\|=1$. Thus, noting that $T_R P_i\in\Omega_{\mathcal{G}}$ for all $i=1,\ldots,D$, we get
\begin{align}
\|T_R {\rm P}\| &= \left\| \sum_{i=1}^D \alpha_i T_R P_i\right\|\leq \sum_{i=1}^D |\alpha_i|\cdot \|T_R P_i\| \nonumber \\
& = \|{\rm P}\|_1 \leq c\cdot \|{\rm P}\|.
\end{align}
This establishes that the operator norm of all $T_R$ with respect to $\|\bullet\|$ (and hence with respect to all other norms) is uniformly bounded. 
Since we have assumed that $\mathcal{G}$ is locally compact, this implies that there is a basis of $\Omega_{\mathcal{G}}^{\reals}$ in which the $T_R$ are orthogonal matrices.

Consider now the {\em evaluation functional} $\delta^{a}_{\vec{x}}:  \Omega^\reals_{\mathcal{G}} \to \reals$;
  namely, the map from the space of black boxes to the particular probability of outcome $a$ given input $\vec{x}$.
It follows that the statistics $\cPr{a}{\vec{x}} = \cPr{a}{R\vec{x}_0} = T_R P^a\!\left(\vec{x}_0\right) = \delta^{a}_{\vec{x}_0} \left(T_R P \right)$.
Since the evaluation functional is a linear map, we then find that the probabilities are given by a linear combination of elements from $T_R$.
For all $\vec{x}\in\mathcal{X}$, we use the same $P$ and the same $\delta^{a}_{\vec{x}_0}$ such that the only element that changes is the representation matrix $T_R$.
\end{proof}

Arguing via harmonic analysis on homogeneous spaces~\cite{Kirillov95}, we expect that Theorem~\ref{thm:Rep} can be extended: it is not only entries of $T_R$ that appear in the probability table $P(a|x_R)$, but, more specifically, generalized spherical harmonics. 
A taste of this appears in Lemma~\ref{lem:Bloch}, but since the formulation of Theorem~\ref{thm:Rep} is sufficient for the purpose of this article, we defer this extension to future work.

\section{$\SO{2}\times\SO{2}$ Bell experiment setting}
\subsection{General form of correlations}
\label{app:SO2xSO2}
\begin{lemma}
\label{thm:SO2Prob}
Consider a bipartite $\SO{2}\!\times\!\SO{2}$-box 
-- i.e.\ Alice and Bob can each choose their inputs as angles $\alpha, \beta \in [0,2\pi)$ -- 
 with local binary outcomes $a\in\{+1,-1\}$ and $b\in\{+1,-1\}$.
Then, the most general joint probability distribution consistent with Theorem~\ref{thm:Rep} is
\begin{align}
\label{eq:GenericProbApp}
\cPr{a,b}{\alpha,\beta} := &\sum_{m=0}^{2J} \sum_{n=-2J}^{2J} c^{ab}_{mn} \cos\left( m \alpha - n \beta\right) \nonumber \\
&\hspace{6em} + s^{ab}_{mn} \sin\left( m \alpha - n \beta\right),
\end{align}
where $J$ is some non-negative integer or half-integer. (Note that this does not yet assume the no-signalling principle.)
\begin{proof}
While the representation $T_R=T_{\alpha,\beta}$ from Theorem~\ref{thm:Rep} acts on a real vector space $V$ of finite dimension $D$, we can also regard it as a representation on the complexification $W=V\oplus iV$. 
Since ${\rm SO}(2)\times {\rm SO}(2)$ is an Abelian group, all its irreducible representations are one-dimensional~\cite{Sepanski07}. 
Thus, we can decompose $W$ as $W=\bigoplus_{j=1}^D W_j$, where each $W_j$ is a one-dimensional invariant subspace on which $T_{\alpha,\beta}$ acts as a complex phase. 
It follows that $T_{\alpha,\beta}=\bigoplus_{j=1}^D \exp(i(m_j\alpha-n_j\beta))$ with suitable integers $m_j,n_j\in\ints$ (to see this, write $T_{\alpha,\beta}$ as a composition of the ${\rm SO}(2)$-representations $T_{\alpha,0}$ and $T_{0,\beta}$). 
Then, due to Theorem~\ref{thm:Rep}, $\cPr{a,b}{\alpha,\beta}$ must be a linear combination of real and imaginary parts of $T_{\alpha,\beta}$, which proves that it is of the form~(\ref{eq:GenericProbApp}).
\end{proof}
\end{lemma}

\subsection{Generic no-signalling correlations}
\label{app:AnyCIsNoSig}
It is well-known {(e.g.\ \cite{PopescuR94})} that the no-signalling principle does not impose any constraints on the form of the correlation function if we have a bipartite box with two outcomes $a,b\in\{+1,-1\}$ each. 
Namely, if $\mathcal{X}, \mathcal{Y}$ denote two arbitrary sets of inputs, 
 given an arbitrary function $C:\mathcal{X}\times\mathcal{Y}\to\reals$ with $-1\leq \corr{x,y} \leq 1$ for all $x,y \in \mathcal{X},\mathcal{Y}$, the simple prescription
\begin{equation}
   \cPr{a,b}{x,y} := \frac{1}{4} + \frac{1}{4} a b \,\corr{x,y}
\end{equation}
generates a valid no-signalling distribution that has $C(x,y)$ as its correlation function. 
It is a simple exercise to check that $C$ is non-negative, normalized and no-signalling, and that $C$ is the correlation function for $\mathrm{P}$.

\subsection{Local hidden variable models for\\ $\SO{2}\times\SO{2}$ settings}
\label{app:LHV}
Generalizing ideas of~\citet{Werner89}, we can show that for noisy enough correlation functions of $\SO{2}\times\SO{2}$ settings,
 we can always construct a LHV model that achieves these correlations.

\begin{lemma}
\label{LemLHV1}
Consider any two-angle function
\small
\begin{equation}
\label{eq:IndexCorrelations}
f(\alpha,\beta)= \sum_{j=1}^N\left[ c_j \cos(m_j\alpha-n_j\beta)+s_j\sin(m_j\alpha-n_j\beta)\right],
\end{equation}
\normalsize
for $(m_j,n_j)\in\ints\times\ints\setminus(0,0)$ (i.e.\ disallowing constant terms).
Without loss of generality\footnote{
These restrictions ensure that the coefficients $c_j$ and $s_j$ are associated with unique trigonometric functions.
} in $f(\alpha,\beta)$,
 we disallow $(m_i,n_i) = (m_j,n_j)$ when $i\neq j$,
 choose $m_j \geq 0$, 
 and if $m_j=0$ then we choose $n_j>0$.
Suppose $-1\leq f(\alpha,\beta)\leq 1$ for all $\alpha,\beta$. 
Then $C(\alpha,\beta):=\gamma\, f(\alpha,\beta)$ is a correlation function that has a LHV model whenever $0\leq\gamma\leq \gamma_N$, where
\begin{equation}
   \gamma_N := \sqrt{\frac 2 N}\max_{0\leq x \leq\pi}\left(\frac x \pi\right)^{N-1}\frac{\sin x}\pi.
\end{equation}
\end{lemma}
\begin{proof}
In this proof, we will express all angles as numbers in the interval $[-\pi,\pi)$. 
Under the inner product $\langle f,g\rangle:=\frac 1 {2\pi^2}\int_{-\pi}^{\pi}\mathrm{d}\alpha\int_{-\pi}^{\pi}\mathrm{d}\beta f(\alpha,\beta)g(\alpha,\beta)$, the set of functions
\begin{equation}
   \cos(m_j\alpha-n_j\beta), \quad\sin(m_j\alpha-n_j\beta)
   \label{eqONS}
\end{equation}
(with $m_j$ and $n_j$ defined as above) is an orthonormal system
(this follows from Schur orthonormality for $\SO{2}\times\SO{2}$, and can be verified by direct integration).
Hence the $L^2$-norm $\|f\|^2 := \langle f, f\rangle$ satisfies
\begin{equation}
\|f\|^2 = \int_{-\pi}^{\pi}\mathrm{d}\alpha\int_{-\pi}^{\pi}\mathrm{d}\beta \frac{f\left(\alpha,\beta\right)^2}{2\pi^2}
=  \sum_{j=1}^N (c_j^2+s_j^2)  
\leq 2,
\end{equation}
since $|f(\alpha,\beta)|\leq 1$ everywhere.

Our goal is to construct a LHV model of the form
\begin{equation}
\cPr{a,b}{\alpha,\beta} = \int_{\Lambda} \mathrm d \mu(\lambda) \; \cPrs{A}{a}{\alpha,\lambda} \cPrs{B}{b}{\beta,\lambda}.
\end{equation}
We will have a hidden variable $\lambda=(\vec\lambda_1,\vec\lambda_2,\ldots,\vec\lambda_N)$, where each $\vec\lambda_j=(\cos\phi_j,\sin\phi_j)\trans$ is independently and uniformly distributed on the unit circle, hence $\mathrm{d}\mu(\lambda)=(2\pi)^{-N} \mathrm{d}\phi_1\ldots\mathrm{d}\phi_N$. 
This measure is invariant under $\SO{2}$ rotations of the individual $\vec\lambda_j$.

We will construct local probabilities that implement the dependence on $\alpha$, $\beta$, $\lambda$ in the following form:
\begin{align}
\cPrs{A}{\pm}{\alpha,\lambda} & 
 = q_A\left(\pm\left|R_{m_1\alpha}\vec\lambda_1,  \ldots, R_{m_N\alpha}\vec\lambda_N \right.\right), \\
\cPrs{B}{\pm}{\beta,\lambda} & 
 = q_B\left(\pm\left|R_{n_1\beta}\vec\lambda_1,\ldots, R_{n_N\beta}\vec\lambda_N \right.\right),
\end{align}
where $q_A$ and $q_B$ are response functions defined in the following way: $q_A(-|\vec\lambda'_1,\ldots,\vec\lambda'_N)=1-q_A(+|\vec\lambda'_1,\ldots,\vec\lambda'_N)$,
\begin{equation}
   q_A(+|\vec\lambda'_1,\ldots,\vec\lambda'_N):=\left\{
      \begin{array}{cl}
      	  1 & \mbox{if }\phi'_j\in[-\xi,\xi]\mbox{ for all }j\\
      	  0 & \mbox{otherwise}
      \end{array}
   \right.
\end{equation}
where $\xi\in (0,\pi)$ is some (small) constant and $\phi'_j$ is the angle such that $\vec\lambda'_j=(\cos\phi'_j,\sin\phi'_j)\trans$. Furthermore,
\begin{equation}
   q_B\left(\pm | \vec\lambda'_1,\ldots,\vec\lambda'_N\right):=\frac 1 2\left(1\pm \frac 1 {\sqrt{2N}} \sum_{j=1}^N \vec b_j\cdot\vec\lambda'_j\right),
\end{equation}
where $\vec b_j:=(c_j,-s_j)\trans$. 
Note that $\sum_j\vec b_j\cdot\vec\lambda'_j=\vec b\cdot\vec\lambda'$, where $\vec b:=\oplus_j \vec b_j$ and $\vec\lambda':=\oplus_j\vec\lambda'_j$. 
But since $|\vec b|^2=\sum_j |\vec b_j|^2=\sum_j(c_j^2+s_j^2)\leq 2$ and $|\lambda'|^2=\sum_j |\vec\lambda'_j|^2=N$, the sum hence is upper-bounded by $\sqrt{2N}$ due to the Cauchy-Schwarz inequality. 
This shows that $q_B$ yields valid probabilities.

We calculate the joint probability distribution obtained in the Bell test scenario:
\begin{align}
\cPr{+,\pm}{\alpha, \beta} &= \int_{\Lambda}\mathrm d \mu(\lambda) \; \cPrs{A}{+}{\alpha,\lambda} \cPrs{B}{\pm}{\beta, \lambda}  \nonumber  \\
	&\hspace{-3em}= \int_{\Lambda}\mathrm d \mu( \lambda) \; q_A(+|R_{m_1\alpha}\vec\lambda_1,\ldots,R_{m_N\alpha}\vec\lambda_N) \nonumber \\
	& \cdot q_B(\pm|R_{n_1\beta}\vec\lambda_1,\ldots,R_{n_N\beta}\vec\lambda_N).
\end{align}
We apply the substitution $\vec\lambda'_j:=R_{m_j\alpha}\vec\lambda_j$ and $\lambda':=(\vec\lambda'_1,\ldots,\vec\lambda'_N)$, noting that this does not change the integral due to our choice of measure:
\begin{align}
\cPr{+,\pm}{\alpha, \beta} & = \int_{\Lambda}\mathrm d \mu( \lambda')\ q_A(+|\vec\lambda'_1,\ldots,\vec\lambda'_N) \nonumber \\
&  \cdot	q_B(\pm|R_{n_1\beta-m_1\alpha}\vec\lambda'_1,\ldots,R_{n_N\beta-m_N\alpha}\vec\lambda'_N).
\end{align}
Due to the definition of $q_A$ and $q_B$, this equals
\begin{equation}
   \int_{-\xi}^{\xi}\frac{\mathrm{d}\phi'_1}{2\pi}\ldots \int_{-\xi}^{\xi}\frac{\mathrm{d}\phi'_N}{2\pi}\frac 1 2 \left(1\pm \frac 1 {\sqrt{2N}}\sum_j f_j(\phi'_j)\right),
\end{equation}
where 
\begin{equation}
f_j(\phi'_j)= c_j \cos(\phi'_j+n_j\beta-m_j\alpha)
-s_j\sin(\phi'_j+n_j\beta-m_j\alpha).
\end{equation}
Noting that
\begin{align}
   \int_{-\xi}^{\xi} f_j(\phi'_j)\mathrm{d}\phi'_j&=2\sin\xi\left[\strut c_j\cos(m_j\alpha-n_j\beta)\right. \nonumber\\
   & \qquad\qquad \left. + s_j \sin(m_j\alpha-n_j\beta)\strut\right],
\end{align}
we can evaluate the integral explicitly, obtaining
\begin{align}
\cPr{+\pm}{\alpha,\beta} = \hspace{-5em} & \nonumber \\
& \frac{1}{2} \left(\frac{\xi}{\pi}\right)^N
\pm \;
\frac{1}{2\sqrt{2N}}\left(\frac{\xi}{\pi}\right)^{N-1} 
\frac{\sin\xi}\pi \, f(\alpha,\beta).
\end{align}
Next, let us look at the other probabilities:
\small
\begin{align}
\cPr{-,\pm}{\alpha,\beta} \hspace{-4.5em} & \hspace{4.5em} = \int_{\Lambda}\mathrm d \mu(\lambda) \cPrs{A}{-}{\alpha,\lambda} \cPrs{B}{\pm}{\beta, \lambda} \nonumber \\
&= \int_{\Lambda}\mathrm d \mu(\lambda) (1-\cPrs{A}{+}{\alpha,\lambda}) \cPrs{B}{\pm}{\beta, \lambda} \nonumber \\
& = -\mathrm{P}(+,\pm|\alpha,\beta)+\int_\Lambda \mathrm{d}\mu(\lambda) q_B(\pm|R_{n_1\beta}\vec\lambda_1,\ldots,R_{n_N\beta}\vec\lambda_N),
\end{align}
\normalsize
and the final integral vanishes on all $\vec b_j\cdot\vec\lambda'_j$-terms of $q_B$, leaving only the constant term $1/2$. 
That is,
\begin{equation}
\cPr{-,\pm}{\alpha,\beta} = \frac{1}{2} -\mathrm{P}(+,\pm|\alpha,\beta).
\end{equation}
These give the correlation function
\begin{equation}
\corr{\alpha,\beta}=\sqrt{\frac{2}{N}} \left(\frac{\xi}{\pi}\right)^{N-1}\frac{\sin\xi}\pi f(\alpha,\beta).
\end{equation}
Finally, we define $\gamma_N$ as the largest admissible prefactor among all possible choices of $\xi$.
\end{proof}

Let us now introduce a constant term:
\begin{lemma}
\label{lem:LHVconstant}
Consider any two-angle correlation function
\small
\begin{equation}
\label{eq:IndexCorrConstant}
C(\alpha,\beta)= c_0+\sum_{j=1}^N\left[ c_j \cos(m_j\alpha-n_j\beta)+s_j\sin(m_j\alpha-n_j\beta)\right],
\end{equation}
\normalsize
where $(m_j,n_j)\in \ints\times\ints\setminus(0,0)$,
 and (as above) without loss of generality we choose $m_j\geq0$ and $n_j>0$ if $m_j=0$,
  and disallow $(m_i,n_i) = (m_j,n_j)$ if $i\neq j$. If
\begin{equation}
   \max_{\alpha,\beta}|C(\alpha,\beta)-c_0|\leq\gamma_N(1-|c_0|)
   \label{eqIneq}
\end{equation}
with constant $\gamma_N$ given by
\begin{equation}
\label{eq:AppGamma}
   \gamma_N=\left\{
      \begin{array}{cl}
      	   \sqrt{2}/\pi & \mbox{if }N=1,\\
      	   0.184375\ldots & \mbox{if }N=2,\\
      	   0.103893\ldots & \mbox{if }N=3,\\
\sqrt{2}e^{-1} \, N^{-3/2} & \mbox{if }N\geq 4,
      \end{array}
   \right.
\end{equation} 
then this correlation function has a local hidden variable model. 
\begin{proof}
First, we obtain the form of $\gamma_N$ by solving the optimization problem of Lemma~\ref{LemLHV1} exactly for $N=1,2,3$, and by substituting $x=\pi\left(1-\frac 1 N\right)$ and using
\begin{equation}
   \left(1-\frac 1 N\right)^{N-1}\sin\left[\pi\left(1-\frac 1 N\right)\right]\geq \frac \pi {Ne} \mbox{ for }N\geq 4.
\end{equation}
We can add the constant function $1/\sqrt{2}$ to the orthonormal system in~(\cref{eqONS}); similar reasoning as in the proof of Lemma~\ref{LemLHV1} shows that $(\sqrt{2}c_0)^2\leq 2$, i.e.\ that $-1\leq c_0\leq 1$, and $|c_0|=1$ is only possible if $C(\alpha,\beta)=c_0$ (i.e.\ with no angle-dependent terms). 
Now consider the case $0\leq c_0<1$. We can write
\begin{equation}
   C(\alpha,\beta)=c_0\mathbf{1}+(1-c_0)f(\alpha,\beta),
\end{equation}
where $\mathbf{1}$ is the constant function that takes the value $1$ on all angles, and $f(\alpha,\beta)=(C(\alpha,\beta)-c_0)/(1-c_0)$ is of the form of the function in Lemma~\ref{LemLHV1}. If inequality~(\ref{eqIneq}) holds, then
\begin{equation}
   \max_{\alpha,\beta}|f(\alpha,\beta)|=\frac 1 {1-c_0}\max_{\alpha,\beta}|C(\alpha,\beta)-c_0|\leq \gamma_N,
\end{equation}
and so Lemma~\ref{LemLHV1} proves that $f(\alpha,\beta)$ is a classical correlation function. Moreover, $\mathbf{1}$ is trivially a classical correlation function, and thus so must be $C(\alpha,\beta)$, which is a convex combination of the two. Then case $-1<c_0<0$ can be treated analogously, using that $-\mathbf{1}$ is a classical correlation function too.
\end{proof}
\end{lemma}

\vspace*{0.5em}
\noindent \textbf{Proof of Theorem 2A.}
Consider an $\SO{2}\times\SO{2}$ box, with a correlation function in the form of \cref{eq:GenericCorr} with maximum (half-)integer $J\neq 0$.
If
\begin{equation}
   \max_{\alpha,\beta}|C(\alpha,\beta)-C_{00}|\leq\gamma_J(1-|C_{00}|),
\end{equation}
where $C_{00}$ is the angle-independent contribution to the correlation function (as in \cref{eq:GenericCorr}),
 and $\gamma_J$ is a given 
\begin{equation}
\gamma_J =  \sqrt{2} e^{-1} \left[4J\left(2J+1\right)\right]^{-\frac{3}{2}},
\end{equation}
then there is a LHV model that accounts for these correlations.
\begin{proof}
This follows as a corollary of Lemma~\ref{lem:LHVconstant}.
We convert between the form of correlations in \cref{eq:GenericCorr} and \cref{eq:IndexCorrConstant}
 by counting the maximum number $N$ of unique terms that could appear for a given positive (half-)integer $J$.
The double sum contributes $(2J+1) (4J+1)$ terms, 
 from which we remove $2J$ cases corresponding to negative $n$ where $m\!=\!0$, 
 and the one completely constant case $m\!=\!n\!=\!0$.
This gives a maximum of $N=4J(2J+1)$.
Since the lowest value ($J=\frac{1}{2}$) already yields $N=4$ unique terms, we only need the final case of \cref{eq:AppGamma},
 and hence the constant
$\gamma_J = \sqrt{2} e^{-1} \left[4J\left(2J+1\right)\right]^{-\frac{3}{2}}$.
\end{proof}

\subsection{Witnessing nonlocality}
\label{app:NonLocal}

Bell inequalities can be {\em chained} by direct addition.
For instance, suppose one takes a CHSH inequality (\cref{eq:CHSH}) with measurements $\{x_1, y_2, x_3, y_4\}$ and a second with measurements $\{x_1, y_4, x_5, y_6\}$.
Adding these together yields
$\big| \corr{x_1,y_2} + \corr{x_3,y_2} + \corr{x_3,y_4} + \corr{x_5,y_4} + \corr{x_5,y_6} - \corr{x_1,y_6} \big| \leq 4$.
This can inductively be done for a set of $N$ measurements ($\frac{N}{2}$ each for Alice and Bob), 
 leading to a chained Bell inequality, known as the {\em Braunstein--Caves inequality} (BCI)~\cite{BraunsteinC90}:
\begingroup
\addtolength{\jot}{-0.5em}
\begin{align}
\label{eq:BCI}
\Big| \corr{x_1, y_2} + \corr{x_3, y_2} + \corr{x_3, y_4} + \cdots & \nonumber \\
 + \corr{x_{N-1}, y_N} - \corr{x_1, y_N} \Big| & \leq N-2.
\end{align}
\endgroup
If such an equation is violated, then no LHV can account for these statistics\footnote{
The BCI can also be directly justified, just as the CHSH inequality. One writes $\sigma_{x_1}\left(\sigma_{y_2}-\sigma_{y_N}\right) + \sigma_{x_3}\left(\sigma_{y_2}+\sigma{y_4}\right) + \ldots + \sigma_{x_{N-1}}\left(\sigma_{y_{N-2}}+\sigma_{y_N}\right)$
for spins $\{\sigma_i \in \{+1,-1\}\}$,
and notes that if $\sigma_{y_2} = \sigma_{y_4} = \ldots = \sigma_{y_N}$, then $\sigma_{y_2} - \sigma_{y_N} = 0$. This bounds the expression to $N-2$. Convex combinations, such as \cref{eq:BCI}, cannot exceed this value.
}.

Recall, \cref{eq:GenericCorr} gives the generic $\SO{2}\times\SO{2}$ correlation function.
If we restrict ourselves to {\em relational} correlations, this amounts to setting $S_{mn}=C_{mn} = 0$ when $m\neq n$, such that the correlation function has a single parameter form
\begin{equation}
\label{eq:RelGeneric}
\corr{\beta-\alpha} = \sum_{m=0}^{2J} C_{m} \cos\left[m\left(\beta- \alpha\right)\right] + S_{m} \sin\left[m\left(\beta- \alpha\right)\right]
\end{equation}
where $J$ is some positive (half-)integer, and $C_m := C_{mm}$, $S_m := S_{mm}$.

\begin{lemma}
\label{lem:Relational}
Consider relational $\SO{2}\times\SO{2}$ correlations (of the form of \cref{eq:RelGeneric}) for finite positive \mbox{(half-)}integer $J$.
If there is some $\Theta_+ \in [0,2\pi)$ such that $C\!\left(\Theta_+\right) = +1$,
 and $\Theta_- \in [0,2\pi)$ where $C\!\left(\Theta_-\right) = -1$,
 then there exists a BCI that demonstrates a Bell violation.
\begin{proof}
We show this by construction.
For even $N$, define 
\begin{equation}
\delta_N := \dfrac{\left(\Theta_- - \Theta_+\right)\mod 2\pi}{N-1}.
\end{equation}
We use the notation ``$x$ mod $2\pi$'' to indicate $x - 2\pi n$ where $n\in\ints$ is chosen such that $x-2\pi n \in [0,2\pi)$, mapping the angle to the principal range.

\begin{figure}[bth]
\includegraphics[width=0.225\textwidth]{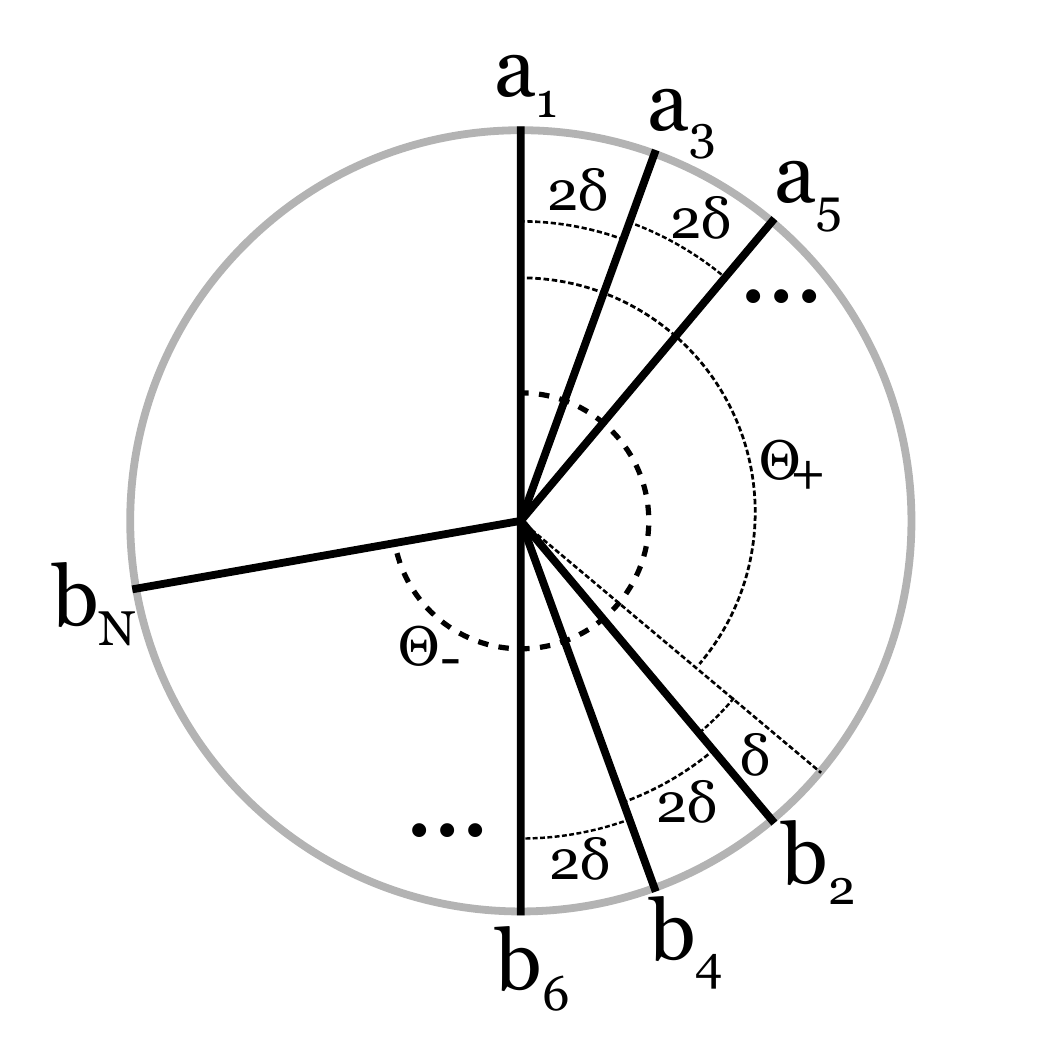}
\caption{
\label{fig:GenericBCI}
\caphead{Measurement angles for generic correlation function.}
The first choice of measurement angles are chosen such that $b_2 - a_1 = \Theta_+ + \delta$.
Subsequent choices then precess by $2\delta$, such that ultimately $b_N-a_1 = \Theta_-$. 
}
\end{figure}

We construct a $N$-measurement BCI, as defined in \cref{eq:BCI}.
Since the correlation function is relational, we write $\corr{\alpha,\beta}$ as the single parameter function $\corr{\beta-\alpha}$,
 and assign the measurement settings:
\begin{align}
\label{eq:RelAssignments}
a_i & = (i-1)\delta_N & \mathrm{for~odd}~ i, \nonumber \\
b_i & = \Theta_+ + (i-1)\delta_N & \mathrm{for~even~} i.
\end{align}
(Illustrated in \cref{fig:GenericBCI}.)
This amounts to setting the arguments of the correlation functions featured in the BCI to
\begin{align}
b_2 - a_1 & = \ldots = b_{2m}-a_{2m-1} = \Theta_+ + \delta_N, \nonumber \\
b_2 - a_3 & = \ldots = b_{2m}-a_{2m+1} = \Theta_+ - \delta_N, \nonumber \\
b_N - a_1 & = \Theta_-,
\end{align}
where equality is taken modulo $2\pi$.
 
 With such assignments, the BCI is then written:
\begin{equation}
\label{eq:moreGenBCI}
\frac{N}{2} C\!\left(\Theta_+ + \delta_N\right) 
+ \left(\frac{N}{2} -1\right) C\!\left(\Theta_+ - \delta_N\right)
- C\left(\Theta_-\right) \leq N-2.
\end{equation}
Recall that $C\left(\Theta_-\right) = -1$.
$C\!\left(\Theta_+\right) = +1$ must be a local maximum, and a finite $J$ allows us to assume the function $C$ is smooth at this point.
Thus, in the limit of small $\delta_N$, $C(\Theta_+\pm\delta_N)=1-k{\delta_N}^2+\order{{\delta_N}^3}$ where $k\geq 0$ is some constant.
We then rewrite \cref{eq:moreGenBCI} as
\begin{align}
\left(N-1\right) \left(1 - k {\delta_N}^2+\order{{\delta_N}^3}\right) + 1 & \leq N-2, \nonumber\\
    N+\order{(N-1)^{-1}}&\leq N-2,
\end{align}
which for large enough $N$ will eventually be violated.
\end{proof}
\end{lemma}

\enlargethispage{\baselineskip}
The above construction can be shown to have some robustness to noise (tolerating a smaller maximum value than $1$).
To show this, we first prove an auxilary lemma:
\begin{lemma}
\label{lem:MaxSecondDeriv}
Consider relational $\SO{2}\times\SO{2}$ correlations (of the form of \cref{eq:RelGeneric}) for finite \mbox{(half-)}integer $J$.
Expressed as a single parameter function $\corr{\beta-\alpha}$, the second derivative is everywhere bounded:
\begin{equation}
\label{eq:SecondDerivBound}
\left| \frac{{\rm d}^2\, \corr{\beta-\alpha} }{{\rm d}\!\left(\beta-\alpha\right)^2} \right| \leq \frac{\sqrt{2} J\left(2J+1\right)\left(4J+1\right)}{3}.
\end{equation}
\begin{proof}
In polar form, the correlation function is:
\begin{equation}
\label{eq:CorrPolar}
\corr{\beta-\alpha} = A_0 + \sum_{m=1}^{2J} A_{m} \cos\left(m\left(\beta- \alpha\right) - \phi_m\right),
\end{equation}
where $A_0=C_0$, $A_m=\sqrt{C_m^2+S_m^2}$, and
\begin{equation}
\phi_m=\left\{\begin{array}{cl}
\arctan(S_m/C_m) & \mbox{if }C_m\neq 0\\
\pi/2 & \mbox{if }C_m=0,S_m\geq 0\\
-\pi/2 & \mbox{if }C_m=0,S_m<0.
\end{array}
\right.
\end{equation}
The second derivative with respect to $\beta-\alpha$ is:
\begin{equation}
\label{eq:SecondDeriv}
\frac{{\rm d}^2}{{\rm d}\!\left(\beta-\alpha\right)^2} \corr{\beta-\alpha} = \sum_{m=1}^{2J} - m^2 A_{m} \cos\left[m\left(\beta- \alpha\right) - \phi_m\right].
\end{equation}
Because $\corr{\beta-\alpha}$ is bounded everywhere to $[-1,1]$, its $L^2$ norm $\|C\|^2 := \frac{1}{2\pi}\int_0^{2\pi} d\theta \,  |C(\theta)|^2$ is bounded within $[0,1]$.
Thus we may determine a maximum value over all $m$ for the amplitude $A_m$.
Using the orthonormality of the functions $1$, $\sqrt{2}\cos(m\theta+a)$, and $\sqrt{2}\cos(n\theta+b)$ under the corresponding inner product $\langle f,g\rangle:=\frac 1 {2\pi}\int_0^{2\pi} d\theta f(\theta)g(\theta)$, we get $\|C\|^2=A_0^2+\frac 1 2 \sum_{m=1}^{2J}A_m^2$. Since this is upper-bounded by $1$, we get $A_m^2\leq 2$ for all $m\geq 1$, and so
\begin{equation}
\label{eq:SecondDerivBound}
\left| \frac{{\rm d}^2 C\!\left(\beta-\alpha\right) }{{\rm d}\!\left(\beta-\alpha\right)^2} \right| \leq \sum_{m=1}^{2J} \sqrt{2} m^2 = \frac{\sqrt{2} J\left(2J+1\right)\left(4J+1\right)}{3}.
\end{equation}%
\end{proof}
\end{lemma}

\begin{lemma}
\label{lem:NoisyRelational}
Consider relational $\SO{2}\times\SO{2}$ correlations (of the form of \cref{eq:RelGeneric}) for finite positive \mbox{(half-)}integer $J$.
Let there be some angle difference $\Theta_+$, where $\corr{\Theta_+} \geq 1 - \varepsilon$, 
 and some other angle difference $\Theta_-$  where $\corr{\Theta_-} \leq 1-\Delta$, with $\varepsilon\geq 0$ and $\Delta>0$. If 
\begin{align}
\label{eq:EpApp}
\varepsilon<- K_J+\sqrt{K_J^2+\frac{\Delta^2}4}\quad = \frac{\Delta^2}{8 K_J}+\mathcal{O}(K_J^{-2}),
\end{align}
where $K_J=\sqrt{2}\pi^2 J(2J+1)(4J+1)/3$, then there will be a BCI that is violated.

\begin{proof}
First, since the correlation function $C$ is continuous, it attains its global maximum at some $\Theta'_+$. Since $C(\Theta'_+)\geq C(\Theta_+)\geq 1-\varepsilon$, the premises of this lemma are also satisfied if $\Theta_+$ is replaced by $\Theta'_+$ -- i.e., we can assume without loss of generality that $C$ attains its global maximum at $\Theta_+$.

With these $\Theta_+$ and $\Theta_-$, we use the prescription in Lemma~\ref{lem:Relational},
 with the angle choices in \cref{eq:RelAssignments} to generate the following BCI, which must hold for all even integers $N\geq 2$ if there exists a LHV model:
\begin{equation}
\frac{N}{2} \corr{\Theta_+ + \delta_N}
+ \left(\frac{N}{2} -1\right) \corr{\Theta_+ - \delta_N}
- \corr{\Theta_-} \leq N-2.
\label{eqBCI}
\end{equation}
Let us write $\delta\Theta := \Theta_- - \Theta_+ \mod 2\pi$, such that $\delta_N = \frac{\delta\Theta}{N-1}$. 
Let $K\in\reals$ be any constant such that $C''(x)\geq -K$ for all $x$; it follows from Lemma~\ref{lem:MaxSecondDeriv} that such $K$ exists, and we will fix $K$ later in accordance with that lemma. Since $\Theta_+$ is a local maximum, we know that $0\geq C''(\Theta_+)\geq -K$, i.e.\ $K\geq 0$. The global bound on the second derivative of $C$ gives us
\begin{equation}
   C(\Theta_+\pm x)\geq C(\Theta_+)-K\frac{x^2}2 \mbox{ for all }x.
\end{equation}
Thus, eq.~(\ref{eqBCI}) implies
\begin{equation}
   (N-1)\left(1-\varepsilon-K\frac{\delta_N^2}2\right)-(1-\Delta)\leq N-2.
\end{equation}
Under what conditions does there exist an even integer $N\geq 2$ such that this inequality is violated, i.e.\ the existence of a LHV model is ruled out? The negation of this inequality can be rearranged into a quadratic equation in $(N-1)$:
\begin{equation}
   \varepsilon(N-1)^2-\Delta(N-1)+\frac K 2 (\delta\Theta)^2<0.
   \label{eqImpliesNonloc}
\end{equation}
If this equation has a solution for some even integer $N$, then the non-existence of a LHV model follows.
If $\varepsilon=0$, then there will always be a solution for large enough $N$, recovering Lemma~\ref{lem:Relational}.
Thus, we here only give further consideration to the case where $\varepsilon>0$.

Since this quadratic function in $(N-1)$ is positive for large values of $\pm (N-1)$, it is necessary for the existence of a negative value that this function has zeroes over the real numbers. The zeroes are
\begin{equation}
   N_\pm-1=\frac{\Delta\pm\sqrt{\Delta^2-2\varepsilon K (\delta\Theta)^2}}{2\varepsilon},
\end{equation}
and so the following inequality is necessary for the existence of a solution of eq.~(\ref{eqImpliesNonloc}):
\begin{equation}
   \Delta^2> 2\varepsilon K (\delta\Theta)^2.	
   \label{eqNecessary}
\end{equation}
If it is satisfied, then the values of $(N-1)_{\pm}$ are well-defined, and we can continue to argue as follows. The quadratic function in eq.~(\ref{eqImpliesNonloc}) is negative for all real numbers $N\in(N_-,N_+)$, where $0<N_-<N_+$. Now, this interval definitely contains an even integer $N$ if $N_+-N_->2$. Since $N_+-N_-=\sqrt{\Delta^2-2\varepsilon K (\delta\Theta)^2}/\varepsilon$, this difference is larger than two if and only if
\begin{equation}
   4\varepsilon^2+2\varepsilon K (\delta\Theta)^2-\Delta^2<0.
\end{equation}
The two solutions of the corresponding quadratic equation are
\begin{equation}
   \varepsilon_{\pm}=\frac{-K(\delta\Theta)^2\pm \sqrt{K^2(\delta\Theta)^4+4\Delta^2}}4.
\end{equation}
They are both real, and $\varepsilon_-<0<\varepsilon_+$. Thus, $\varepsilon<\varepsilon_+$ implies a suitable solution of eq.~(\ref{eqImpliesNonloc}), i.e.\ rules out the existence of a LHV model.

In fact, if $\varepsilon<\varepsilon_+$, then we automatically get
\begin{equation}
4\varepsilon < -K(\delta\Theta)^2+K(\delta\Theta)^2\sqrt{1+\frac{4\Delta^2}{K^2(\delta\Theta)^4}}
\leq \frac{2\Delta^2}{K(\delta\Theta)^2},
\end{equation}
i.e.\ eq.~(\ref{eqNecessary}) is automatically satisfied. Now, considering $\varepsilon_+$ as a function in $\delta\Theta$, this function is decreasing for $\delta\Theta>0$. Since $\delta\Theta\leq 2\pi$, $\varepsilon<\varepsilon_+(2\pi)$ implies $\varepsilon<\varepsilon_+(\delta\Theta)\equiv \varepsilon_+$. Thus, the inequality
\begin{equation}
   \varepsilon<\frac{-K(2\pi)^2+\sqrt{K^2(2\pi)^4+4\Delta^2}}4
\end{equation}
implies a violation of a BCI. The statement of the lemma now follows from taking the value of $K$ from Lemma~\ref{lem:MaxSecondDeriv}, and by substituting $K_J:=\pi^2 K$.
\end{proof}
\end{lemma}

This has consequence for generic (possibly non-relational) $\SO{2}\times\SO{2}$ settings.

\vspace{0.5em} \noindent \textbf{Proof of Theorem 2B.}
{\em
Consider $\SO{2}\times\SO{2}$ correlations $C$ for finite maximum \mbox{(half-)}integer $J$.
Let $C_{\rm rel}$ the relational core of $C$ (that is, the function of the form \cref{eq:RelGeneric} formed by only including terms of \cref{eq:GenericCorr} where $m=n$). 
Let there be some angle difference $\Theta_+$, where $C_{\rm rel}\!\left(\Theta_+\right) \geq 1 - \varepsilon$, 
 and some other angle difference $\Theta_-$  where $C_{\rm rel}\!\left(\Theta_-\right) \leq 1-\Delta$, with $\varepsilon\geq 0$ and $\Delta>0$. If 
\begin{align}
\varepsilon<- K_J+\sqrt{K_J^2+\frac{\Delta^2}4} \; = \; \frac{\Delta^2}{8 K_J}+\mathcal{O}(K_J^{-2}),
\end{align}
where $K_J=\sqrt{2}\pi^2 J(2J+1)(4J+1)/3$, then there will be a BCI for the (possibly non-relational) correlation function $C$ that is violated.
}
\begin{proof}
Subtracting the ``non-relational'' parts of $\corr{\alpha,\beta}$ is equivalent to performing the following integration:
\begin{equation}
C_{\rm rel}\!\left(\alpha, \beta\right) = \frac{1}{2\pi} \int_0^{2\pi} d\phi \; \corr{\alpha+\phi, \beta+\phi}
\end{equation}
This may be directly verified by noting that terms of the form $\cos\left(m\alpha- n\beta + \left(m-n\right)\phi\right)$ and 
$\sin\left(m\alpha- n\beta + \left(m-n\right)\phi\right)$ individually integrate to $0$ over $\phi$ except when $m=n$.
This allows us to interpret taking the relational core of a correlation function as mixing $C$ over many settings offset by a shared uniform random angle.

It then follows from the convexity of Bell inequalities that if the BCI implied by Lemma~\ref{lem:NoisyRelational} for the relational core is violated ``on average'' for this mixture of settings,
 there must be at least one single set of input settings that also results in that BCI being violated.
\end{proof}

\subsection{Necessity of a bound on $J$}
\label{app:Necessity}
We will now show that our protocol for witnessing nonlocality does not work if we simply drop the assumption that $J$ is finite.

A correlation function $C(\alpha,\beta)$ has a LHV model if and only if there exists a variable $\lambda\in\Lambda$, distributed via some $P_{\Lambda}(\lambda)$, and a family of local response functions $C_A(\alpha,\lambda) := \sum_{a\in\{-1,+1\}} a\, P_A(a|\alpha,\lambda)$ and $C_B(\beta,\lambda) := \sum_{b\in\{-1,+1\}} b\, P_B(b|\beta,\lambda)$ such that
\begin{equation}
   C(\alpha,\beta)=\int_\Lambda d\lambda P_\Lambda(\lambda) C_A(\alpha,\lambda) C_B(\beta,\lambda).
   \label{DefLocalC}
\end{equation}
Suppose that there are two angles $\theta'_+,\theta'_-$ such that our protocol gives correlation values $C_{\rm rel}(\theta'_{\pm})$ very close to $\pm 1$. 
If this is the \emph{only} experimental syndrome, without further assumptions on the form of the correlation function (in particular, without any assumption on $J$ as explained in the main text), then this experimental behavior can always be reproduced to arbitrary accuracy by local hidden variables. 
Namely, $\theta'_\pm$ can be arbitrarily well approximated by angles $\theta_\pm$ that satisfy the premises of the following lemma:

\begin{lemma}
Suppose that $\theta_+,\theta_-\in[0,2\pi)$ are such that $\theta_--\theta_+=\frac m n \pi$, where $n\in\mathbb{N}$ and $m\leq n$ is an odd integer. 
Then there exists a local relational correlation function $C(\alpha,\beta)\equiv C(\alpha-\beta)$ such that
\begin{align}
   C(\theta_+)=+1 \mbox{ and } C(\theta_-)=-1.
\end{align}
\end{lemma}
\begin{proof}
We set $\Lambda=[0,2\pi)$ and $P_\Lambda(\lambda)=\frac{1}{2\pi}$ -- the uniform measure on this interval. 
Without loss of generality, assume that $\Theta_- > \Theta_+$ and $\Theta_+ = 0$ (we can choose our local coordinates $\alpha$, $\beta$ to make this the case).
Define $C_A(x,\lambda) = C_B(x,\lambda) := f(x+\lambda)$, where $f:\mathbb{R}\to\{-1,+1\}$ is the $2\pi$-periodic extension of
\begin{align}
   f(x) & := \left\{\begin{array}{cl}
      +1 & \mbox{if }\frac x \pi \in \left[0,\frac 1 n\right) \cup \left[\frac 2 n,\frac 3 n\right)\cup\ldots\cup\left[\frac{2n-2}n,\frac{2n-1}n\right) \\
      -1 & \mbox{for all other }x\in[0,2\pi).
\end{array}
\right.
\end{align}
That is, $f$ is a square-wave function of period $\frac{2\pi}{n}$.
This $f$ is piecewise continuous and satisfies $f\left(x+\frac\pi n\right)=-f(x)$ for all $x$. Thus, $f\left(x+\frac 3 n \pi\right)=-f\left(x+\frac 2 n \pi\right)=f\left(x+\frac 1 n \pi\right)=-f(x)$ for all $x$, and in particular, by induction, $f\left(x+\frac m n \pi\right)=-f(x)$ for all $x$ since $m$ is odd by assumption. Now, defining $C(\alpha,\beta)$ as in~(\ref{DefLocalC}), this correlation function is relational, since
\begin{align}
C(\alpha+x,\beta+x)&=\frac 1 {2\pi} \int_0^{2\pi} f(\alpha+x+\lambda)f(\beta+x+\lambda)\, d\lambda \nonumber \\
&= \frac 1 {2\pi} \int_0^{2\pi} f(\alpha+\lambda')f(\beta+\lambda')\,d\lambda' \nonumber \\
&= C(\alpha,\beta)
\end{align}
by substitution and due to the $(2\pi)$-periodicity of $f$. Furthermore,
\begin{align}
C(\theta_+)&=C\left(0,0\right)=\frac 1 {2\pi} \int_0^{2\pi}\left(f(\lambda)\right)^2 \, d\lambda=1,\nonumber \\
C(\theta_-)&= C\left(\frac m n \pi,0\right)=\frac 1 {2\pi}\int_0^{2\pi} f\left(\frac m n \pi+\lambda\right)f(\lambda)\,d\lambda\nonumber \\
&=-\frac 1 {2\pi} \int_0^{2\pi}\left(f(\lambda)\right)^2 \, d\lambda=-1.
\end{align}
\end{proof}
Therefore, simply following our protocol but relaxing our assumption on $J$ (while not imposing any other assumptions) cannot be sufficient to certify nonlocality.

\section{Characterizing quantum correlations}
\label{app:Vector}
Let us consider black boxes that have a particularly simple transformation behaviour under rotations:
\begin{definition}[Transforming fundamentally]
\label{def:TransFund}
Consider an $\SO{d}$-box $\cPr{a}{x}$, where $d\geq 2$.
Let $x_0\in\mathcal{X}$. We say that this box \emph{transforms fundamentally} under rotations if for all $x \in \mathcal X$ and all $R_x \in \mathcal G$ with $R_x x_0 = x$ one finds
\begin{equation}
\label{eq:TransFund}
\cPr{a}{x}\equiv \cPr{a}{R_x x_0}=c_0^{a}+\sum_{i,j=1}^d (R_x)_{i,j} c_{i,j}^{a},
\end{equation}
 where $\left(R_x\right)_{i,j}$ is the fundamental matrix representation of $R_x\in\mathcal{G}$, and $c_0^{a}, c_{i,j}^{a}$ are constants independent of $x$ and $R_x$.
\end{definition}
Equivalently, a black box transforms fundamentally if the corresponding representation $R\mapsto T_R$ from Theorem~\ref{thm:Rep} can be chosen as a direct sum of copies of the trivial and the fundamental representations of $\mathcal{G}=\SO{d}$.
Since $\mathcal{G}$ is transitive on $\mathcal{X}$,
 the existence of the above representation is independent of the particular $x_0$: 
  any alternative $x_0' \in \mathcal{X}$ satisfies $x_0' = Sx_0$ for some $S\in\mathcal{G}$, 
  satisfying the above with $R' = RS^{-1}$.

\begin{lemma}
\label{lem:Bloch}
Suppose that $\mathcal{X}=S^{d-1}$ (the unit sphere), and $\cPr{a}{\vec{x}}$ transforms fundamentally under rotations in $\mathcal{G}=\SO{d}$. 
Then, 
\begin{equation}
\label{eq:GenericBloch}
 \cPr{a}{\vec{x}}=c_0^{a}+\vec c^{\,a}\cdot \vec x,
\end{equation}
where constants $c_0^{a}\in\reals$ and $\vec c^{\,a}\in\reals^d$ satisfy
\begin{equation}
\label{eq:GBNormConst}
\sum_{a\in\mathcal{A}} c_0^{a} = 1,\quad \sum_{a\in\mathcal{A}} \vec{c}^{\,a} = \vec{0} 
\end{equation}
such that for all $a$, $c_0^{a}\geq 0$ and $|\vec{c}^{\,a}| \leq \min\left(c_0^a,1-c_0^a\right)$.

Conversely, if a black box has this form, then it transforms fundamentally under rotations.

In other words, an $\SO{d}$-box transforms fundamentally if and only if $\cPr{a}{\vec{x}}$ is affine-linear in $\vec{x}$ (and non-negativity and normalization of probabilities holds).

\begin{proof}
Set $\vec x_0:=\vec e_1=(1,0,\ldots,0)\trans$. 
Fix some choice of rotations $\vec x \mapsto R_{\vec x}$ with $R_{\vec x}\vec x_0=\vec x$. Consider the stabilizer subgroup $\mathcal{G}_{\vec e_1}:=\{R\in\SO{d}\,\,|\,\, R\vec e_1=\vec e_1\}$:
\begin{align}
   \mathcal{G}_{\vec e_1} 
   &=\left\{\left(\begin{array}{cc} 1 & \vec{0}\trans \\ \vec 0 & T\end{array}\right)\,\,|\,\, T\in\SO{d-1}\right\}.
\end{align}
The fact that this group is isomorphic to $\SO{d-1}$ is precisely due to the fact that our set of inputs is the homogeneous space $\mathcal{X}=\SO{d}/\SO{d-1}$, i.e.\ the $(d-1)$-sphere.
 For $d\geq 3$, we have $\int_{{\rm SO}(d-1)} T\, {\mathrm d}T=0$, and thus
\begin{align}
   \int_{\mathcal{G}_{\vec e_1}}S\, {\mathrm d}S=Q:=\left(\begin{array}{cc} 1 & \vec{0}\trans \\ \vec 0 & \mathbf{0}\end{array}\right).
\end{align}
Since $R_{\vec x}S \vec e_1=\vec x$ for every $S\in\mathcal{G}_{\vec e_1}$, every rotation matrix $R_{\vec x}S$ can be substituted for $R_x$ into \cref{def:TransFund}. 
Thus $\cPr{a}{\vec{x}} = c_0^{a}+\sum_{i,j=1}^d (R_{\vec x}S)_{i,j}c^{a}_{i,j}$.
By taking the average over all $S\in\mathcal{G}_{\vec e_1}$ according to the Haar measure, we get
\begin{equation}
P(a|\vec x)= c_0^{a}+\sum_{i,j=1}^d \left(R_{\vec x}\int_{\mathcal{G}_{\vec e_1}}S\,{\mathrm d}S\right)_{i,j}c^{a}_{i,j}.
\end{equation}
But $R_{\vec x}Q=(\vec x,\vec 0,\ldots,\vec 0)$, i.e.\ a matrix with first column equal to $\vec x$ and all further columns equal to zero. This proves that $P(a|\vec x)$ is affine-linear in $\vec x$ as claimed, in the case $d\geq 3$.

Now consider the case $d=2$. Here, $\vec x=(x_1,x_2)\trans$, and there is a unique choice of $R_{\vec x}$, namely $R_{\vec x}=\left(\begin{array}{cc} x_1 & -x_2 \\ x_2 & x_1\end{array}\right)$. Then, $P(a|\vec x)$ being affine-linear in $R_{\vec x}$ is equivalent to being affine-linear in $\vec{x}$.

From normalization, $\sum_a \cPr{a}{\vec{x}} = 1\; \forall \vec{x}\in\mathcal{X}$,
 which by transitivity of $\mathcal{G}$ on $\mathcal{X}$ can be re-written $\sum_a \cPr{a}{R\vec{x}_0} = 1\;\forall R\in\mathcal{G}$.
Suppose we take the Haar average of $\mathcal{G}$ over both sides of this constraint:
\begin{align}
\int_{\SO{d}} \hspace{-1.5em} dR\, \sum_a \cPr{a}{R\vec{x}_0} & 
 = \sum_a \left( c_0^{a}+\vec{c}^{\,a}\cdot \int_{\SO{d}}\hspace{-1.5em} dR\,R \vec x \right) \nonumber \\
& = \sum_a c_0^{a} =  1,
\end{align}
where we have used $\int_{\SO{d}} R dR = 0$.
Since $\sum_a \cPr{a}{R\vec{x}_0} = 1$ for each individual $R\in\mathcal{G}$,
 we have $\left(\sum_a \vec{c}^{\,a} \right)\cdot \left(R \vec{x}_0\right) = 0$.
Since by transitivity $\dim\!\left[\,{\rm span}\!\left(\{R\vec{x}_0\}_{R\in\mathcal{G}}\right)\right] = d$,
 it follows that $\sum_a \vec{c}^{\,a} = \vec{0}$.

For any $\vec{x}$, one may find some $R\in\SO{d}$ such that $R\vec{x} = -\vec{x}$ (since both $\vec{x}$ and $-\vec{x} \in \sphere{d-1}$ and $\SO{d}$ is transitive on $\sphere{d-1}$).
Hence, for the black box $\cPr{a}{\vec{x}}$, there is always another black box $\cPr{a}{-\vec{x}}$ such that the average statistics of these two boxes is given by $\frac{1}{2}\left[\cPr{a}{\vec{x}}+\cPr{a}{-\vec{x}}\right] = c_0^a$.
Clearly, then $c_0^a \geq 0$.
Finally, from the definition of the dot product, $\min_{\vec{x}\in\sphere{d-1}} \left( \vec{c}\cdot\vec{x} \right) = -|\vec{c}\,|$. 
Thus, if $|\vec{c}^{\,a}| > c_0^a$, there will be some $\vec{x}$ such that $\cPr{a}{\vec{x}} = c_0^a - |\vec{c}^{\,a}| < 0$, which is not a valid probability.  
Similarly $\max_{\vec{x}\in\sphere{d-1}} \left( \vec{c}\cdot\vec{x} \right) = |\vec{c}\,|$, so if $|\vec{c}^{\,a}| > 1 - c_0^a$, there will be some $\vec{x}$ such that $\cPr{a}{\vec{x}} = c_0^a + |\vec{c}^{\,a}|\geq 1$.
Hence $|\vec{c}^{\,a}| \leq \min\left(c_0^a,1-c_0^a\right)$.

The converse follows from the transitivity of $\SO{d}$ on $\sphere{d-1}$: any $\vec{x}$ can be expressed as $R\vec{x}_0$ for some fixed $\vec{x}_0$ and $R\in\SO{d}$. Thus \cref{eq:GenericBloch} can be written $\cPr{a}{\vec{x}} = c_0^a + \vec{c}^{\,a}\cdot R\vec{x}_0$ which has the form of \cref{eq:TransFund}.
\end{proof}
\end{lemma}

We can formally define the concept of an ``unbiased'' black box where if the input orientation is randomized, no particular outcome is preferred:
\begin{definition}[Unbiased]
\label{def:Unbiased}
Consider a $\mathcal{G}$-box $\cPr{a}{x}$ for some compact group $\mathcal{G}$. 
We say that this box is {\em unbiased}  if the Haar average of $\cPr{a}{Rx}$ over $R\in\mathcal{G}$ is the same for all $a\in\mathcal{A}$.
\end{definition}
It follows from normalization that if a black box transforms fundamentally under rotations and is unbiased, it may be written in the form $\cPr{a}{\vec{x}} = \frac{1}{|\mathcal{A}|} + \vec{c}^{\,a}\cdot\vec{x}$.

We extend both these definitions to the local parts of a bipartite system by considering the {\em conditional boxes}
\begin{align}
P_A^{b,\vec{y}}(a|\vec{x})&:=\cPr{a,b}{\vec{x},\vec{y}}/\cPrs{B}{b}{\vec{y}}\quad (b,\vec{y}\mbox{ fixed}), \label{Eq:CondBoxA}\\
P_B^{a,\vec{x}}(b|\vec{y})&:=\cPr{a,b}{\vec{x},\vec{y}}/\cPrs{A}{a}{\vec{x}}\quad (a,\vec{x}\mbox{ fixed}),
\end{align}
defined whenever $\cPrs{B}{b}{\vec{y}}>0$ and $\cPrs{A}{a}{\vec{x}}>0$ respectively.

A {\em conditional box} can be thought of as the black box Alice has if she is told Bob's measurement and outcome. 
(This is in contrast to a {\em marginal} black box, which quantifies Alice's statistics when she knows nothing of Bob's measurement or outcome.)
No-signalling implies the existence of well-defined marginal boxes $\cPrs{B}{b}{\vec{y}}$ and $\cPrs{A}{a}{\vec{x}}$.

\begin{definition}
We say that a no-signalling bipartite box $\cPr{a,b}{\vec{x},\vec{y}}$ \emph{transforms fundamentally locally} (is \emph{locally unbiased}) if all conditional boxes transform fundamentally (are unbiased).
\end{definition}

The next two lemmas show that these properties are preserved by convex combinations of boxes.
\begin{lemma}
\label{lem:ConvexTLF}
Let $\{\cPrs{i}{a,b}{\vec{x},\vec{y}}\}_{i=1,\ldots N}$ be a set of no-signalling bipartite black boxes that transform fundamentally locally.
Any convex combination $\cPr{a,b}{\vec{x},\vec{y}} := \sum_i \lambda_i \cPrs{i}{a,b}{\vec{x},\vec{y}}$ where all $\lambda_i \geq 0$ and $\sum_i\lambda_i=1$ also transforms fundamentally locally.
\begin{proof}
First, we calculate the marginal distribution:
$\cPrs{A}{a}{\vec{x}} \ :=\  \sum_b \cPr{a,b}{\vec{x},\vec{y}} = \sum_b \sum_i \lambda_i \cPrs{i}{a,b}{\vec{x},\vec{y}}  = \sum_i \lambda_i \sum_b \cPrs{i}{a,b}{\vec{x},\vec{y}} = \sum_i \lambda_i \cPrs{A,i}{a}{\vec{x}}$.
Similarly, $\cPrs{B}{b}{\vec{y}} = \sum_i \lambda_i \cPrs{B,i}{b}{\vec{y}}$.
First, we note that $\cPrs{B}{b}{\vec{y}}=0$ only if $\cPrs{B,i}{b}{\vec{y}}=0$ for all $i$. In this case, the combined conditional box is undefined, and there is nothing to prove.
Thus, we may proceed with the case that $\cPrs{B}{b}{\vec{y}}>0$.
 
With $\cPr{a,b}{\vec{x},\vec{y}} = \sum_i \lambda_i \cPrs{i}{a,b}{\vec{x},\vec{y}} = {\sum_{i}}' \lambda_i \cPrs{B,i}{b}{\vec{y}} {\rm P}^{b,\vec{y}}_{A,i}\!\left(a|\vec{x}\right)$ we obtain
\begin{align}
{\rm P}^{b,\vec{y}}_{A}\!\left(a|\vec{x}\right) & =  {\sum_{i}}'  \dfrac{\lambda_i \cPrs{B,i}{b}{\vec{y}}}{\cPrs{B}{b}{\vec{y}}} {\rm P}^{b,\vec{y}}_{A,i}\!\left(a|\vec{x}\right) \label{Eq:CondBoxConvMix}
\end{align}
Here, ${\sum_{i}}'$ denotes a sum over all those $i$ for which $P_{B,i}(b|\vec y)>0$. These are exactly the $i$ for which ${\rm P}^{b,\vec{y}}_{A,i}\!\left(a|\vec{x}\right)$ is well-defined.
Meanwhile, ${\sum_{i}}' \frac{\lambda_i \cPrs{B,i}{b}{\vec{y}}}{\cPrs{B}{b}{\vec{y}}} = 1$, and hence we may define $\mu_i := \frac{\lambda_i \cPrs{B,i}{b}{\vec{y}}}{\cPrs{B}{b}{\vec{y}}}$ for those $i$ with $P_{B,i}(b|\vec y)>0$, and $\mu_i=0$ for all other $i$, such that $\mu_i\!\geq\!0$ and moreover {\mbox{$\sum_i \mu_i={\sum_{i}}'\mu_i=1$}}. 
Thus, the new conditional box is itself a convex combination of the constituent conditional boxes.
A similar convex combination can be found for ${\rm P}^{a,\vec{x}}_{B}\!\left(b|\vec{y}\right)$.

Since the constituent conditional boxes transform fundamentally, 
 from Lemma~\ref{lem:Bloch} we write ${\rm P}^{b,\vec{y}}_{A,i}\!\left(a|\vec{x}\right) =c_{i,0}^{(b,\vec{y}),a}+\vec c_i^{\,(b,\vec{y}),a}\cdot \vec x$. 
Then, 
${\rm P}^{b,\vec{y}}_{A}\!\left(a|\vec{x}\right)
  = \sum_i \mu_i \left( c_{i,0}^{(b,\vec{y}),a}+\vec{c}_i^{\,(b,\vec{y}),a} \cdot \vec{x} \right)
= c_0^{(b,\vec{y}),a} + \vec{c}^{\,(b,\vec{y}),a} \cdot \vec{x}$
where 
$c_0^{(b,\vec{y}),a} := \sum \mu_i c_{i,0}^{(b,\vec{y}),a}$ and $\vec{c}^{\,(b,\vec{y}),a} := \sum \mu_i \vec{c}_i^{\,(b,\vec{y}),a}$.
By the converse part of Lemma~\ref{lem:Bloch}, ${\rm P}^{b,\vec{y}}_{A}\!\left(a|\vec{x}\right)$ is thus a valid black box that transforms fundamentally.
The same argument holds for Bob's conditional boxes. 
\end{proof}
\end{lemma}

\begin{lemma}
\label{lem:ConvexLU}
Let $\{\cPrs{i}{a,b}{\vec{x},\vec{y}}\}_{i=1,\ldots N}$ be a set of non-signalling black boxes that are locally unbiased with respect to $\mathcal{G}$. 
Any convex combination $\cPr{a,b}{\vec{x},\vec{y}} := \sum_i \lambda_i \cPrs{i}{a,b}{\vec{x},\vec{y}}$ where all $\lambda_i \geq 0$ and $\sum_i\lambda_i=1$ is also locally unbiased with respect to $\mathcal{G}$.
\begin{proof}
In the proof of Lemma~\ref{lem:ConvexTLF} we have seen that there exists a probability distribution $\{ \mu_i \}_i $ such that ${\rm P}^{b,\vec{y}}_{A}\!\left(a|\vec{x}\right)  = \sum_i \mu_i {\rm P}^{b,\vec{y}}_{A,i}\!\left(a|\vec{x}\right)$. Thus we find 
\begin{align}
	\int{ \rm P}^{b,\vec{y}}_{A}\!\left(a|R \vec{x}\right) \mathrm d R & =  \sum_i \mu_i \int {\rm P}^{b,\vec{y}}_{A,i}\!\left(a|R \vec{x}\right) \mathrm d R \nonumber \\
	& = \frac{\sum_i \mu_i}{|\mathcal A|} = \frac{1}{|\mathcal A|}.
\end{align}
Likewise holds for Bob's conditional boxes, and hence, $\cPr{a,b}{\vec{x},\vec{y}}$ is also locally unbiased.
\end{proof}
\end{lemma}

\begin{lemma}
\label{lem:LFLU}
Consider a bipartite black box with inputs $\mathcal{X}=\mathcal{Y}=\sphere{d-1}$ and binary outcomes $\mathcal{A}=\mathcal{B}=\{+1,-1\}$.
If this box transforms fundamentally locally and is locally unbiased, then it describes quantum correlations.
\begin{proof}
From Lemma~\ref{lem:Bloch}, binary-outcome conditional boxes that transform fundamentally and are unbiased can be written:
\begin{align}
{\rm P}_A^{b,\vec{y}}(a|\vec{x}) & := \left(\begin{array}{c} \frac{1}{2} \\ a \vec{c}^{\,(b,\vec{y})} \end{array}\right) \cdot \left(\begin{array}{c} 1 \\ \vec{x} \end{array}\right)
=
\left(\begin{array}{c} \frac{1}{2} \\ \vec{c}^{\,(b,\vec{y})} \end{array}\right) \cdot \left(\begin{array}{c} 1 \\ a \vec{x} \end{array}\right)
, \\
{\rm P}_B^{a,\vec{x}}(b|\vec{y}) & := \left(\begin{array}{c} \frac{1}{2} \\ b \vec{c}^{\,(a,\vec{x})} \end{array}\right) \cdot \left(\begin{array}{c} 1 \\ \vec{y} \end{array}\right) = \left(\begin{array}{c} \frac{1}{2} \\ \vec{c}^{\,(a,\vec{x})} \end{array}\right) \cdot \left(\begin{array}{c} 1 \\ b \vec{y} \end{array}\right),
\end{align}
such that the joint probability distribution is given by:
\begin{align}
\label{eq:JointProbBilinearA}
\cPr{a,b}{\vec{x},\vec{y}} & = {\rm P}_B\!\left(b\middle|\vec{y}\right)\left(\begin{array}{c} \frac{1}{2}\\ \vec{c}^{\,(b,\vec{y})} \end{array}\right) \cdot \left(\begin{array}{c} 1 \\ a\vec{x} \end{array}\right), \\
\label{eq:JointProbBilinearB}
\cPr{a,b}{\vec{x},\vec{y}} & = {\rm P}_A\!\left(a\middle|\vec{x}\right)\left(\begin{array}{c} \frac{1}{2}\\ \vec{c}^{\,(a,\vec{x})} \end{array}\right) \cdot \left(\begin{array}{c} 1 \\ b\vec{y} \end{array}\right).
\end{align}

This defines a map $\tilde{\omega}_{AB}$ acting on $\vec{e}_{a,\vec{x}} := \left(1, a\vec{x}\right)\trans$ and $\vec{e}_{b,\vec{y}}:= \left(1,b\vec{y}\right)\trans$,
 such that $\tilde{\omega}_{AB}(\vec{e}_{a,\vec{x}},\vec{e}_{b,\vec{y}}) =\cPr{a,b}{\vec{x},\vec{y}}$.
Moreover, ${\rm span}\!\left(\vec{e}_{a,\vec{x}}\right) ={\rm span}\!\left(\vec{e}_{b,\vec{y}}\right) = \reals^{d+1}$,
 and so this function has a unique bilinear extension $\omega_{AB}: \reals^{d+1}\times\reals^{d+1} \to \reals$.
 
The set of non-negative linear combinations of $\vec{e}_{a,\vec{x}}$ define a {\em positive Euclidean cone} $A^+ \subset \reals^{d+1}$ , whose extremal rays are the non-negative multiples of $\left(1,\vec{z}\right)\trans$ for $\vec{z}\in\sphere{d-1}$.
We may then define an {\em Archimedean order unit} (AOU) \cite{KleinmannOSW13}, $\vec{u}:=(2,0,\ldots,0)\trans \in \reals^{d+1}$
 and define an AOU-space $(\reals^{d+1},A^+,\vec{u})$.
An identical AOU-space $(\reals^{d+1},B^+,\vec{u})$ can be defined using the non-negative linear combinations of $\vec{e}_{b,\vec{y}}$.
 
Now, we shall employ a result from \citet{KleinmannOSW13} (generalizing a result by \citet{BarnumBBEW10}) that pertains to bilinear maps on positive Euclidean cones.
If a bilinear map $\omega_{AB}$ on such cones is both {\em unital} and {\em positive}, then there exists a bipartite quantum system $\rho_{AB}$ and a map from each point $\vec{a}\in A^+$, $\vec{b}\in B^+$ onto local quantum POVM elements $M_a$, $M_b$ such that $\omega_{AB}( \vec{a} , \vec{b}) = {\rm tr} \left(\rho_{AB} M_a\otimes M_b\right)$.
 
We show that $\omega_{AB}$ satisfies these conditions.
First, for any given $a$, $\vec{x}$ (likewise $b$, $\vec{y}$), it can be seen that
\begin{equation}
\vec{e}_{+a,\vec{x}}+\vec{e}_{-a,\vec{x}} = \vec{u} = \vec{e}_{+b,\vec{y}}+\vec{e}_{-b,\vec{y}},
\end{equation}
and hence 
$\omega_{AB}\left(\vec{u},\vec{u}\right)  = \sum_{a,b} \cPr{a,b}{\vec{x},\vec{y}} = 1$,
 which means that $\omega_{AB}$ is {\em unital}.
Next, since 
 every $\vec{p}\in A^+$ can be written as a non-negative linear combination of finitely many $e_{a,\vec{x}}$ (likewise for $\vec{q}\in B^+$), then $\omega_{AB}(\vec{p},\vec{q})\geq 0$ for all $\vec{p},\vec{q}\in A^+,B^+$, showing that $\omega_{AB}$ is \emph{positive}.
Hence, $\omega_{AB}$ can be realised by a quantum system, and $\cPr{a,b}{\vec{x},\vec{y}}$ is a quantum behaviour.
\end{proof}
\end{lemma}

The premise of local unbiasedness cannot be removed: if we only demand that a box transforms fundamentally locally, then it can generate correlations that are disallowed by quantum theory. 
To see this, let $P_0$ be any non-signalling $(2,2,2)$-behaviour (for example, a PR-box), and define
\begin{align}
\cPr{a,b}{\vec{x},\vec{y}} &:= \lambda_A\lambda_B P_0(a,b|0,0)+\lambda_A\bar\lambda_B P_0(a,b|0,1) \nonumber \\
&\quad +\bar{\lambda}_A \lambda_B P_0(a,b|1,0) + \bar\lambda_A\bar\lambda_B P_0(a,b|1,1),
\end{align}
where $a,b\in\{-1,+1\}$, $\vec x,\vec y\in S^{d-1}$, $\lambda_A:=\frac 1 2(1+x_1)$, $\bar\lambda_A:=1-\lambda_A$, $\lambda_B:=\frac 1 2 (1+y_1)$, and $\bar\lambda_B:=1-\lambda_B$. 
If $d=3$, for example, this describes a situation in which two possible local inputs $x,y\in\{0,1\}$ are encoded into the first Bloch vector component of a qubit, the qubits are locally measured, and the measurement results $x,y\in\{0,1\}$ are input into the original box $P_0$.
This defines a valid non-signalling box, and the linear dependence of the outcome probabilities on $\vec x$ (resp.\ $\vec y$) demonstrate, via Lemma~\ref{lem:Bloch}, that $P$ transforms fundamentally locally. However, it reproduces $P_0$ via $P_0(a,b|r,t)=\cPr{a,b}{\vec{x}_r,\vec{y}_t}$, where $\vec x_0=\vec y_0=(1,0,\ldots,0)$ and $\vec x_1=\vec y_1=(-1,0,\ldots,0)$.

Now we show a converse statement, so that we get an exact classification of the quantum $(2,2,2)$-behaviours:
 i.e.\ the family of probabilities obtained in quantum theory during a two party Bell test, where each agent has two choices of input and obtains one of two outcomes.

\begin{lemma}
\label{lem:extremeQ}
Let $d\geq 2$. 
Then all extremal quantum (2,2,2)-behaviours can be realized by locally unbiased $\SO{d}$-boxes that transform fundamentally locally with $\mathcal{X}_A=\mathcal{X}_B=\sphere{d-1}$; the two settings (inputs) correspond to two choices of  directions.
\begin{proof}
It has been shown~\cite{Cirelson80,Masanes05,TonerV06,NavascuesGHA15} 
 that the extremal quantum $(2,2,2)$-behaviours can be realised by rank-$1$ projective measurements on two-qubit pure states.
Any extremal quantum $(2,2,2)$-behaviour $\cPr{a,b}{r,t}$ can then be written in the form $\cPr{a,b}{r,t} = \tr[\rho_{AB} (E^{(a)}_r \otimes F^{(b)}_t)]$ where $\rho$ is a pure state of two qubits and $E^{(a)}_r$ and $F^{(b)}_t$ are qubit rank-$1$ projectors, $a,b \in \{-1,+1\}$ and $r,t \in \{1,2\}$. 
We shall show that there exists a non-signalling $\SO{d}\times\SO{d}$-box $\cPr{a,b}{\vec{x}, \vec{y}}$ that transforms fundamentally locally, is locally unbiased, and has choices $\vec{x}_r, \vec{y}_t$ such that $\cPr{a,b}{\vec{x}_r, \vec{y}_t} = \cPr{a,b}{r,t}$.

Write $E := E_1^{(+1)}$ and $\tilde{E} := E_2^{(+1)}$. 
As $\tr{E} = 1$, its expansion in the Pauli operator basis is of the form $E = \frac{1}{2}\left(\id_2 + E_x\sigma_x + E_y\sigma_y + E_z \sigma_z\right)$, and the associated Bloch vector $\vec{x}'_1 := (E_x, E_y, E_z)\trans$ has Euclidean norm $1$. 
Let $\vec{x}'_2$ be the Bloch vector similarly associated with $\tilde{E}$. 
By changing the local bases unitarily, we can ensure that $\vec{x}'_1=(1,0,0)\trans$ and $\vec{x}'_2=(\cos\theta,\sin\theta,0)\trans$, where $0\leq\theta<2\pi$. 
Similarly, we can define the Bloch vectors $\vec y'_t$ for $t=1,2$ via the rank-$1$ projections $F_t^{(b)}$.

Define a linear map $\Pi:\reals^d\to\reals^3$ in the following way. 
If $d=2$, set $\Pi(v_1,v_2)\trans:=(v_1,v_2,0)\trans$; if $d\geq 3$, set $\Pi(v_1,\ldots,v_d)\trans:=(v_1,v_2,v_3)\trans$, which is an orthogonal projection (and the identity if $d=3$). 
Furthermore, for $\vec v\in\reals^3$, define $E_{\vec v}:=\frac 1 2(\id_2+\vec v \cdot\vec\sigma)$, which is positive-semidefinite whenever $|\vec v|\leq 1$. 
Consider the non-signalling $\SO{d}\times\SO{d}$-box
\begin{equation}
\cPr{a,b}{\vec{x},\vec{y}}:=\tr\left[\rho_{AB} E_{a\Pi\vec x}\otimes E_{b\Pi\vec y}\right].
\end{equation}
Since $\left|a\Pi\vec{x}\right|\leq \left|\vec{x}\right|=1$ for $\vec{x}\in S^{d-1}$ (and similarly for $b\Pi\vec y$), this defines a valid (quantum) behaviour. 
The conditional boxes are
\begin{equation}
   P_A^{b,\vec y}(a|\vec x)=\frac{\frac{1}{2} \tr\left[\rho_{AB}(\id+a\vec x \cdot(\Pi\trans \vec\sigma))\otimes E_{b\Pi\vec y} \right]}{\tr(\rho_B E_{b\Pi\vec y})}.
\end{equation}
This expression yields well-defined probabilities by construction, and it is affine-linear in $\vec x$. 
Analogous statements hold for the other conditional boxes. Thus, according to Lemma~\ref{lem:Bloch}, $\cPr{a,b}{\vec{x},\vec{y}}$ transforms fundamentally locally. Furthermore, averaging the above conditional box uniformly over $\vec x$ replaces $\vec x$ by zero and annihilates all dependence on $a$; hence this box is locally unbiased.

Let $\vec x_r\in\reals^d$ be the vector whose first two components are the first two components of $\vec x'_r$, and all other $(d-2)$ components are zero; define $\vec y_t$ analogously. 
Then $\cPr{a,b}{r,t}=\cPr{a,b}{\vec{x}_r,\vec{y}_t}$.
\end{proof}
\end{lemma}

\vspace*{0.5em}
\noindent {\bf Proof of Theorem~\ref{thm:222quantum}.}
{\em Let $d\geq 2$.
The quantum $(2,2,2)$-behaviours are exactly those that can be realised by binary-outcome bipartite {$\SO{d}\!\times\!\SO{d}$-}boxes that transform fundamentally locally and are locally unbiased, restricted to two choices of input direction per party per box, and statistically mixed via shared randomness.
}
\begin{proof}
Lemma~\ref{lem:LFLU} tells us that ``$(2,\sphere{d-1},2)$--behaviours'' that transform fundamentally locally, and are locally unbiased, can be realised by local measurements on a bipartite quantum system.
If we restrict our choice of inputs from the full $\sphere{d-1}$ freedom to just two choices of orientation per party,
 then these will be $(2,2,2)$--behaviours, and since they can be realised by a quantum system, they are quantum $(2,2,2)$--behaviours.

The other direction follows from Lemma~\ref{lem:extremeQ}: all \emph{extremal} quantum $(2,2,2)$-behaviours can be realised by restricting binary-outcome bipartite {$\SO{d}\!\times\!\SO{d}$-}boxes, transforming fundamentally locally and being locally unbiased, to two possible input directions per party. 
Additional shared randomness allows the two parties to generate all statistical mixtures of these behaviours, yielding all further quantum $(2,2,2)$-behaviours.
\end{proof}

Theorem~\ref{thm:222quantum} cannot hold for all $d\geq 2$ without allowing shared randomness. 
For example, suppose that $d=3$, then the proof of Lemma~\ref{lem:LFLU} shows that all correlations realizable with binary-outcome bipartite {$\SO{3}\!\times\!\SO{3}$-}boxes that transform fundamentally locally and are locally unbiased can be realized via unital positive bilinear forms on the positive semidefinite qubit cone. 
Consequently, the result by \citet{BarnumBBEW10} implies that all these correlations can also be realized via POVMs on ordinary two-qubit quantum state space. 
However, \citet{DonohueW15} (extending results by \citet{PalV09}) have shown that the set of $(2,2,2)$-behaviours realizable on two qubits via POVMs is not convex, and thus not equal to the convex set of quantum $(2,2,2)$-behaviours.


\begin{thebibliography}{55}%
\makeatletter
\providecommand \@ifxundefined [1]{%
 \@ifx{#1\undefined}
}%
\providecommand \@ifnum [1]{%
 \ifnum #1\expandafter \@firstoftwo
 \else \expandafter \@secondoftwo
 \fi
}%
\providecommand \@ifx [1]{%
 \ifx #1\expandafter \@firstoftwo
 \else \expandafter \@secondoftwo
 \fi
}%
\providecommand \natexlab [1]{#1}%
\providecommand \enquote  [1]{``#1''}%
\providecommand \bibnamefont  [1]{#1}%
\providecommand \bibfnamefont [1]{#1}%
\providecommand \citenamefont [1]{#1}%
\providecommand \href@noop [0]{\@secondoftwo}%
\providecommand \href [0]{\begingroup \@sanitize@url \@href}%
\providecommand \@href[1]{\@@startlink{#1}\@@href}%
\providecommand \@@href[1]{\endgroup#1\@@endlink}%
\providecommand \@sanitize@url [0]{\catcode `\\12\catcode `\$12\catcode
  `\&12\catcode `\#12\catcode `\^12\catcode `\_12\catcode `\%12\relax}%
\providecommand \@@startlink[1]{}%
\providecommand \@@endlink[0]{}%
\providecommand \url  [0]{\begingroup\@sanitize@url \@url }%
\providecommand \@url [1]{\endgroup\@href {#1}{\urlprefix }}%
\providecommand \urlprefix  [0]{URL }%
\providecommand \Eprint [0]{\href }%
\providecommand \doibase [0]{http://dx.doi.org/}%
\providecommand \selectlanguage [0]{\@gobble}%
\providecommand \bibinfo  [0]{\@secondoftwo}%
\providecommand \bibfield  [0]{\@secondoftwo}%
\providecommand \translation [1]{[#1]}%
\providecommand \BibitemOpen [0]{}%
\providecommand \bibitemStop [0]{}%
\providecommand \bibitemNoStop [0]{.\EOS\space}%
\providecommand \EOS [0]{\spacefactor3000\relax}%
\providecommand \BibitemShut  [1]{\csname bibitem#1\endcsname}%
\let\auto@bib@innerbib\@empty
\bibitem [{\citenamefont {Einstein}\ \emph {et~al.}(1935)\citenamefont
  {Einstein}, \citenamefont {Podolsky},\ and\ \citenamefont
  {Rosen}}]{EinsteinPR35}%
  \BibitemOpen
  \bibfield  {author} {\bibinfo {author} {\bibfnamefont {A.}~\bibnamefont
  {Einstein}}, \bibinfo {author} {\bibfnamefont {B.}~\bibnamefont {Podolsky}},
  \ and\ \bibinfo {author} {\bibfnamefont {N.}~\bibnamefont {Rosen}},\
  }\bibfield  {title} {\enquote {\bibinfo {title} {{Can Quantum-Mechanical
  Description of Physical Reality Be Considered Complete?}}}\ }\href {\doibase
  10.1103/PhysRev.47.777} {\bibfield  {journal} {\bibinfo  {journal} {Physical
  Review}\ }\textbf {\bibinfo {volume} {47}},\ \bibinfo {pages} {777--780}
  (\bibinfo {year} {1935})}\BibitemShut {NoStop}%
\bibitem [{\citenamefont {Bell}(1964)}]{Bell64}%
  \BibitemOpen
  \bibfield  {author} {\bibinfo {author} {\bibfnamefont {J.~S.}\ \bibnamefont
  {Bell}},\ }\bibfield  {title} {\enquote {\bibinfo {title} {{On the Einstein
  Podolsky Rosen paradox}},}\ }\href {\doibase
  10.1103/PhysicsPhysiqueFizika.1.195} {\bibfield  {journal} {\bibinfo
  {journal} {Physics Physique Fizika}\ }\textbf {\bibinfo {volume} {1}},\
  \bibinfo {pages} {195--200} (\bibinfo {year} {1964})}\BibitemShut {NoStop}%
\bibitem [{\citenamefont {Clauser}\ \emph {et~al.}(1969)\citenamefont
  {Clauser}, \citenamefont {Horne}, \citenamefont {Shimony},\ and\
  \citenamefont {Holt}}]{ClauserHSH69}%
  \BibitemOpen
  \bibfield  {author} {\bibinfo {author} {\bibfnamefont {J.}~\bibnamefont
  {Clauser}}, \bibinfo {author} {\bibfnamefont {M.}~\bibnamefont {Horne}},
  \bibinfo {author} {\bibfnamefont {A.}~\bibnamefont {Shimony}}, \ and\
  \bibinfo {author} {\bibfnamefont {R.}~\bibnamefont {Holt}},\ }\bibfield
  {title} {\enquote {\bibinfo {title} {{Proposed Experiment to Test Local
  Hidden-Variable Theories}},}\ }\href {\doibase 10.1103/PhysRevLett.23.880}
  {\bibfield  {journal} {\bibinfo  {journal} {Physical Review Letters}\
  }\textbf {\bibinfo {volume} {23}},\ \bibinfo {pages} {880--884} (\bibinfo
  {year} {1969})}\BibitemShut {NoStop}%
\bibitem [{\citenamefont {Aspect}\ \emph {et~al.}(1982)\citenamefont {Aspect},
  \citenamefont {Grangier},\ and\ \citenamefont {Roger}}]{AspectGR82}%
  \BibitemOpen
  \bibfield  {author} {\bibinfo {author} {\bibfnamefont {A.}~\bibnamefont
  {Aspect}}, \bibinfo {author} {\bibfnamefont {P.}~\bibnamefont {Grangier}}, \
  and\ \bibinfo {author} {\bibfnamefont {G.}~\bibnamefont {Roger}},\ }\bibfield
   {title} {\enquote {\bibinfo {title} {{Experimental Realization of
  Einstein-Podolsky-Rosen-Bohm Gedankenexperiment: A New Violation of Bell's
  Inequalities}},}\ }\href {\doibase 10.1103/PhysRevLett.49.91} {\bibfield
  {journal} {\bibinfo  {journal} {Physical Review Letters}\ }\textbf {\bibinfo
  {volume} {49}},\ \bibinfo {pages} {91--94} (\bibinfo {year}
  {1982})}\BibitemShut {NoStop}%
\bibitem [{\citenamefont {Hensen}\ \emph {et~al.}(2015)\citenamefont {Hensen},
  \citenamefont {Bernien}, \citenamefont {Dr{\'{e}}au}, \citenamefont
  {Reiserer}, \citenamefont {Kalb}, \citenamefont {Blok}, \citenamefont
  {Ruitenberg}, \citenamefont {Vermeulen}, \citenamefont {Schouten},
  \citenamefont {Abell{\'{a}}n}, \citenamefont {Amaya}, \citenamefont
  {Pruneri}, \citenamefont {Mitchell}, \citenamefont {Markham}, \citenamefont
  {Twitchen}, \citenamefont {Elkouss}, \citenamefont {Wehner}, \citenamefont
  {Taminiau},\ and\ \citenamefont {Hanson}}]{Hensen15}%
  \BibitemOpen
  \bibfield  {author} {\bibinfo {author} {\bibfnamefont {B.}~\bibnamefont
  {Hensen}}, \bibinfo {author} {\bibfnamefont {H.}~\bibnamefont {Bernien}},
  \bibinfo {author} {\bibfnamefont {A.~E.}\ \bibnamefont {Dr{\'{e}}au}},
  \bibinfo {author} {\bibfnamefont {A.}~\bibnamefont {Reiserer}}, \bibinfo
  {author} {\bibfnamefont {N.}~\bibnamefont {Kalb}}, \bibinfo {author}
  {\bibfnamefont {M.~S.}\ \bibnamefont {Blok}}, \bibinfo {author}
  {\bibfnamefont {J.}~\bibnamefont {Ruitenberg}}, \bibinfo {author}
  {\bibfnamefont {R.~F.~L.}\ \bibnamefont {Vermeulen}}, \bibinfo {author}
  {\bibfnamefont {R.~N.}\ \bibnamefont {Schouten}}, \bibinfo {author}
  {\bibfnamefont {C.}~\bibnamefont {Abell{\'{a}}n}}, \bibinfo {author}
  {\bibfnamefont {W.}~\bibnamefont {Amaya}}, \bibinfo {author} {\bibfnamefont
  {V.}~\bibnamefont {Pruneri}}, \bibinfo {author} {\bibfnamefont {M.~W.}\
  \bibnamefont {Mitchell}}, \bibinfo {author} {\bibfnamefont {M.}~\bibnamefont
  {Markham}}, \bibinfo {author} {\bibfnamefont {D.~J.}\ \bibnamefont
  {Twitchen}}, \bibinfo {author} {\bibfnamefont {D.}~\bibnamefont {Elkouss}},
  \bibinfo {author} {\bibfnamefont {S.}~\bibnamefont {Wehner}}, \bibinfo
  {author} {\bibfnamefont {T.~H.}\ \bibnamefont {Taminiau}}, \ and\ \bibinfo
  {author} {\bibfnamefont {R.}~\bibnamefont {Hanson}},\ }\bibfield  {title}
  {\enquote {\bibinfo {title} {{Loophole-free Bell inequality violation using
  electron spins separated by 1.3 kilometres}},}\ }\href {\doibase
  10.1038/nature15759} {\bibfield  {journal} {\bibinfo  {journal} {Nature}\
  }\textbf {\bibinfo {volume} {526}},\ \bibinfo {pages} {682--686} (\bibinfo
  {year} {2015})}\BibitemShut {NoStop}%
\bibitem [{\citenamefont {Brunner}\ \emph {et~al.}(2014)\citenamefont
  {Brunner}, \citenamefont {Cavalcanti}, \citenamefont {Pironio}, \citenamefont
  {Scarani},\ and\ \citenamefont {Wehner}}]{BrunnerCPSW14}%
  \BibitemOpen
  \bibfield  {author} {\bibinfo {author} {\bibfnamefont {N.}~\bibnamefont
  {Brunner}}, \bibinfo {author} {\bibfnamefont {D.}~\bibnamefont {Cavalcanti}},
  \bibinfo {author} {\bibfnamefont {S.}~\bibnamefont {Pironio}}, \bibinfo
  {author} {\bibfnamefont {V.}~\bibnamefont {Scarani}}, \ and\ \bibinfo
  {author} {\bibfnamefont {S.}~\bibnamefont {Wehner}},\ }\bibfield  {title}
  {\enquote {\bibinfo {title} {{Bell nonlocality}},}\ }\href {\doibase
  10.1103/RevModPhys.86.419} {\bibfield  {journal} {\bibinfo  {journal}
  {Reviews of Modern Physics}\ }\textbf {\bibinfo {volume} {86}},\ \bibinfo
  {pages} {419--478} (\bibinfo {year} {2014})}\BibitemShut {NoStop}%
\bibitem [{\citenamefont {Mayers}\ and\ \citenamefont {Yao}(1998)}]{MayersY98}%
  \BibitemOpen
  \bibfield  {author} {\bibinfo {author} {\bibfnamefont {D.}~\bibnamefont
  {Mayers}}\ and\ \bibinfo {author} {\bibfnamefont {A.}~\bibnamefont {Yao}},\
  }\bibfield  {title} {\enquote {\bibinfo {title} {{Quantum Cryptography with
  Imperfect Apparatus}},}\ }in\ \href {\doibase 10.1109/SFCS.1998.743501}
  {\emph {\bibinfo {booktitle} {Proceedings of the 39th Annual Symposium on
  Foundations of Computer Science}}}\ (\bibinfo  {publisher} {IEEE},\ \bibinfo
  {address} {Palo Alto},\ \bibinfo {year} {1998})\BibitemShut {NoStop}%
\bibitem [{\citenamefont {Barrett}\ \emph {et~al.}(2005)\citenamefont
  {Barrett}, \citenamefont {Hardy},\ and\ \citenamefont {Kent}}]{BarrettHK05}%
  \BibitemOpen
  \bibfield  {author} {\bibinfo {author} {\bibfnamefont {J.}~\bibnamefont
  {Barrett}}, \bibinfo {author} {\bibfnamefont {L.}~\bibnamefont {Hardy}}, \
  and\ \bibinfo {author} {\bibfnamefont {A.}~\bibnamefont {Kent}},\ }\bibfield
  {title} {\enquote {\bibinfo {title} {{No Signaling and Quantum Key
  Distribution}},}\ }\href {\doibase 10.1103/PhysRevLett.95.010503} {\bibfield
  {journal} {\bibinfo  {journal} {Physical Review Letters}\ }\textbf {\bibinfo
  {volume} {95}},\ \bibinfo {pages} {010503} (\bibinfo {year}
  {2005})}\BibitemShut {NoStop}%
\bibitem [{\citenamefont {Colbeck}\ and\ \citenamefont
  {Renner}(2012)}]{ColbeckR12}%
  \BibitemOpen
  \bibfield  {author} {\bibinfo {author} {\bibfnamefont {R.}~\bibnamefont
  {Colbeck}}\ and\ \bibinfo {author} {\bibfnamefont {R.}~\bibnamefont
  {Renner}},\ }\bibfield  {title} {\enquote {\bibinfo {title} {{Free randomness
  can be amplified}},}\ }\href {\doibase 10.1038/nphys2300} {\bibfield
  {journal} {\bibinfo  {journal} {Nature Physics}\ }\textbf {\bibinfo {volume}
  {8}},\ \bibinfo {pages} {450--453} (\bibinfo {year} {2012})}\BibitemShut
  {NoStop}%
\bibitem [{\citenamefont {Vazirani}\ and\ \citenamefont
  {Vidick}(2014)}]{VaziraniV14}%
  \BibitemOpen
  \bibfield  {author} {\bibinfo {author} {\bibfnamefont {U.}~\bibnamefont
  {Vazirani}}\ and\ \bibinfo {author} {\bibfnamefont {T.}~\bibnamefont
  {Vidick}},\ }\bibfield  {title} {\enquote {\bibinfo {title} {{Fully
  Device-Independent Quantum Key Distribution}},}\ }\href {\doibase
  10.1103/PhysRevLett.113.140501} {\bibfield  {journal} {\bibinfo  {journal}
  {Physical Review Letters}\ }\textbf {\bibinfo {volume} {113}},\ \bibinfo
  {pages} {140501} (\bibinfo {year} {2014})}\BibitemShut {NoStop}%
\bibitem [{\citenamefont {Tura}\ \emph {et~al.}(2014)\citenamefont {Tura},
  \citenamefont {Augusiak}, \citenamefont {Sainz}, \citenamefont
  {V{\'{e}}rtesi}, \citenamefont {Lewenstein},\ and\ \citenamefont
  {Ac{\'{i}}n}}]{TuraASVLA14}%
  \BibitemOpen
  \bibfield  {author} {\bibinfo {author} {\bibfnamefont {J.}~\bibnamefont
  {Tura}}, \bibinfo {author} {\bibfnamefont {R.}~\bibnamefont {Augusiak}},
  \bibinfo {author} {\bibfnamefont {A.~B.}\ \bibnamefont {Sainz}}, \bibinfo
  {author} {\bibfnamefont {T.}~\bibnamefont {V{\'{e}}rtesi}}, \bibinfo {author}
  {\bibfnamefont {M.}~\bibnamefont {Lewenstein}}, \ and\ \bibinfo {author}
  {\bibfnamefont {A.}~\bibnamefont {Ac{\'{i}}n}},\ }\bibfield  {title}
  {\enquote {\bibinfo {title} {{Quantum nonlocality. Detecting nonlocality in
  many-body quantum states.}}}\ }\href {\doibase 10.1126/science.1247715}
  {\bibfield  {journal} {\bibinfo  {journal} {Science (New York, N.Y.)}\
  }\textbf {\bibinfo {volume} {344}},\ \bibinfo {pages} {1256--8} (\bibinfo
  {year} {2014})}\BibitemShut {NoStop}%
\bibitem [{\citenamefont {Schmied}\ \emph {et~al.}(2016)\citenamefont
  {Schmied}, \citenamefont {Bancal}, \citenamefont {Allard}, \citenamefont
  {Fadel}, \citenamefont {Scarani}, \citenamefont {Treutlein},\ and\
  \citenamefont {Sangouard}}]{SchmiedBAFSTS16}%
  \BibitemOpen
  \bibfield  {author} {\bibinfo {author} {\bibfnamefont {R.}~\bibnamefont
  {Schmied}}, \bibinfo {author} {\bibfnamefont {J.-D.}\ \bibnamefont {Bancal}},
  \bibinfo {author} {\bibfnamefont {B.}~\bibnamefont {Allard}}, \bibinfo
  {author} {\bibfnamefont {M.}~\bibnamefont {Fadel}}, \bibinfo {author}
  {\bibfnamefont {V.}~\bibnamefont {Scarani}}, \bibinfo {author} {\bibfnamefont
  {P.}~\bibnamefont {Treutlein}}, \ and\ \bibinfo {author} {\bibfnamefont
  {N.}~\bibnamefont {Sangouard}},\ }\bibfield  {title} {\enquote {\bibinfo
  {title} {{Bell correlations in a Bose-Einstein condensate.}}}\ }\href
  {\doibase 10.1126/science.aad8665} {\bibfield  {journal} {\bibinfo  {journal}
  {Science (New York, N.Y.)}\ }\textbf {\bibinfo {volume} {352}},\ \bibinfo
  {pages} {441--4} (\bibinfo {year} {2016})}\BibitemShut {NoStop}%
\bibitem [{\citenamefont {Pearl}(2000)}]{Pearl00}%
  \BibitemOpen
  \bibfield  {author} {\bibinfo {author} {\bibfnamefont {J.}~\bibnamefont
  {Pearl}},\ }\href@noop {} {\emph {\bibinfo {title} {{Causality: Models,
  reasoning, and inference}}}}\ (\bibinfo  {publisher} {Cambridge University
  Press},\ \bibinfo {year} {2000})\BibitemShut {NoStop}%
\bibitem [{\citenamefont {Nakahara}(2003)}]{Nakahara03}%
  \BibitemOpen
  \bibfield  {author} {\bibinfo {author} {\bibfnamefont {M.}~\bibnamefont
  {Nakahara}},\ }\href@noop {} {\emph {\bibinfo {title} {{Geometry, Topology
  and Physics}}}}\ (\bibinfo  {publisher} {IOP Publishing},\ \bibinfo {year}
  {2003})\BibitemShut {NoStop}%
\bibitem [{\citenamefont {Brunner}\ \emph {et~al.}(2008)\citenamefont
  {Brunner}, \citenamefont {Pironio}, \citenamefont {Ac\'{i}n}, \citenamefont
  {Gisin}, \citenamefont {M{\'{e}}thot},\ and\ \citenamefont
  {Scarani}}]{BrunnerPAGMS08}%
  \BibitemOpen
  \bibfield  {author} {\bibinfo {author} {\bibfnamefont {N.}~\bibnamefont
  {Brunner}}, \bibinfo {author} {\bibfnamefont {S.}~\bibnamefont {Pironio}},
  \bibinfo {author} {\bibfnamefont {A.}~\bibnamefont {Ac\'{i}n}}, \bibinfo
  {author} {\bibfnamefont {N.}~\bibnamefont {Gisin}}, \bibinfo {author}
  {\bibfnamefont {A.~A.}\ \bibnamefont {M{\'{e}}thot}}, \ and\ \bibinfo
  {author} {\bibfnamefont {V.}~\bibnamefont {Scarani}},\ }\bibfield  {title}
  {\enquote {\bibinfo {title} {{Testing the Dimension of Hilbert Spaces}},}\
  }\href {\doibase 10.1103/PhysRevLett.100.210503} {\bibfield  {journal}
  {\bibinfo  {journal} {Physical Review Letters}\ }\textbf {\bibinfo {volume}
  {100}},\ \bibinfo {pages} {210503} (\bibinfo {year} {2008})}\BibitemShut
  {NoStop}%
\bibitem [{\citenamefont {Paw{\l}owski}\ and\ \citenamefont
  {Brunner}(2011)}]{PawlowskiB11}%
  \BibitemOpen
  \bibfield  {author} {\bibinfo {author} {\bibfnamefont {M.}~\bibnamefont
  {Paw{\l}owski}}\ and\ \bibinfo {author} {\bibfnamefont {N.}~\bibnamefont
  {Brunner}},\ }\bibfield  {title} {\enquote {\bibinfo {title}
  {{Semi-device-independent security of one-way quantum key distribution}},}\
  }\href {\doibase 10.1103/PhysRevA.84.010302} {\bibfield  {journal} {\bibinfo
  {journal} {Physical Review A}\ }\textbf {\bibinfo {volume} {84}},\ \bibinfo
  {pages} {010302} (\bibinfo {year} {2011})}\BibitemShut {NoStop}%
\bibitem [{\citenamefont {Bennett}\ and\ \citenamefont
  {Brassard}(1984)}]{BennettB84}%
  \BibitemOpen
  \bibfield  {author} {\bibinfo {author} {\bibfnamefont {C.~H.}\ \bibnamefont
  {Bennett}}\ and\ \bibinfo {author} {\bibfnamefont {G.}~\bibnamefont
  {Brassard}},\ }\bibfield  {title} {\enquote {\bibinfo {title} {{Quantum
  cryptography: Public key distribution and coin tossing}},}\ }in\ \href@noop
  {} {\emph {\bibinfo {booktitle} {Proceedings of IEEE International Conference
  on Computers, Systems and Signal Processing}}},\ Vol.~\bibinfo {volume} {3}\
  (\bibinfo {year} {1984})\ pp.\ \bibinfo {pages} {175--179}\BibitemShut
  {NoStop}%
\bibitem [{\citenamefont {Ac\'{i}n}\ \emph {et~al.}(2006)\citenamefont
  {Ac\'{i}n}, \citenamefont {Gisin},\ and\ \citenamefont {Masanes}}]{AcinGM06}%
  \BibitemOpen
  \bibfield  {author} {\bibinfo {author} {\bibfnamefont {A.}~\bibnamefont
  {Ac\'{i}n}}, \bibinfo {author} {\bibfnamefont {N.}~\bibnamefont {Gisin}}, \
  and\ \bibinfo {author} {\bibfnamefont {Ll.}\ \bibnamefont {Masanes}},\
  }\bibfield  {title} {\enquote {\bibinfo {title} {{From Bell's Theorem to
  Secure Quantum Key Distribution}},}\ }\href {\doibase
  10.1103/PhysRevLett.97.120405} {\bibfield  {journal} {\bibinfo  {journal}
  {Physical Review Letters}\ }\textbf {\bibinfo {volume} {97}},\ \bibinfo
  {pages} {120405} (\bibinfo {year} {2006})}\BibitemShut {NoStop}%
\bibitem [{\citenamefont {Nielsen}\ and\ \citenamefont
  {Chuang}(2000)}]{NielsenC00}%
  \BibitemOpen
  \bibfield  {author} {\bibinfo {author} {\bibfnamefont {M.~A.}\ \bibnamefont
  {Nielsen}}\ and\ \bibinfo {author} {\bibfnamefont {I.~L.}\ \bibnamefont
  {Chuang}},\ }\href@noop {} {\emph {\bibinfo {title} {{Quantum Computation and
  Quantum Information}}}}\ (\bibinfo  {publisher} {Cambridge University
  Press},\ \bibinfo {address} {Cambridge},\ \bibinfo {year} {2000})\BibitemShut
  {NoStop}%
\bibitem [{\citenamefont {Popescu}\ and\ \citenamefont
  {Rohrlich}(1994)}]{PopescuR94}%
  \BibitemOpen
  \bibfield  {author} {\bibinfo {author} {\bibfnamefont {S.}~\bibnamefont
  {Popescu}}\ and\ \bibinfo {author} {\bibfnamefont {D.}~\bibnamefont
  {Rohrlich}},\ }\bibfield  {title} {\enquote {\bibinfo {title} {{Quantum
  nonlocality as an axiom}},}\ }\href {\doibase 10.1007/BF02058098} {\bibfield
  {journal} {\bibinfo  {journal} {Found. Phys.}\ }\textbf {\bibinfo {volume}
  {24}},\ \bibinfo {pages} {379--385} (\bibinfo {year} {1994})}\BibitemShut
  {NoStop}%
\bibitem [{\citenamefont {Werner}(1989)}]{Werner89}%
  \BibitemOpen
  \bibfield  {author} {\bibinfo {author} {\bibfnamefont {R.~F.}\ \bibnamefont
  {Werner}},\ }\bibfield  {title} {\enquote {\bibinfo {title} {{Quantum states
  with Einstein-Podolsky-Rosen correlations admitting a hidden-variable
  model}},}\ }\href {\doibase 10.1103/PhysRevA.40.4277} {\bibfield  {journal}
  {\bibinfo  {journal} {Physical Review A}\ }\textbf {\bibinfo {volume} {40}},\
  \bibinfo {pages} {4277--4281} (\bibinfo {year} {1989})}\BibitemShut {NoStop}%
\bibitem [{\citenamefont {Augusiak}\ \emph {et~al.}(2014)\citenamefont
  {Augusiak}, \citenamefont {Demianowicz},\ and\ \citenamefont
  {Ac{\'{i}}n}}]{AugusiakDA14}%
  \BibitemOpen
  \bibfield  {author} {\bibinfo {author} {\bibfnamefont {R.}~\bibnamefont
  {Augusiak}}, \bibinfo {author} {\bibfnamefont {M.}~\bibnamefont
  {Demianowicz}}, \ and\ \bibinfo {author} {\bibfnamefont {A.}~\bibnamefont
  {Ac{\'{i}}n}},\ }\bibfield  {title} {\enquote {\bibinfo {title} {{Local
  hidden–variable models for entangled quantum states}},}\ }\href {\doibase
  10.1088/1751-8113/47/42/424002} {\bibfield  {journal} {\bibinfo  {journal}
  {Journal of Physics A: Mathematical and Theoretical}\ }\textbf {\bibinfo
  {volume} {47}},\ \bibinfo {pages} {424002} (\bibinfo {year}
  {2014})}\BibitemShut {NoStop}%
\bibitem [{\citenamefont {{\.{Z}}ukowski}(1993)}]{Zukowski93}%
  \BibitemOpen
  \bibfield  {author} {\bibinfo {author} {\bibfnamefont {M.}~\bibnamefont
  {{\.{Z}}ukowski}},\ }\bibfield  {title} {\enquote {\bibinfo {title} {{Bell
  theorem involving all settings of measuring apparatus}},}\ }\href {\doibase
  10.1016/0375-9601(93)90002-H} {\bibfield  {journal} {\bibinfo  {journal}
  {Physics Letters A}\ }\textbf {\bibinfo {volume} {177}},\ \bibinfo {pages}
  {290--296} (\bibinfo {year} {1993})}\BibitemShut {NoStop}%
\bibitem [{\citenamefont {Sen(De)}\ \emph {et~al.}(2002)\citenamefont
  {Sen(De)}, \citenamefont {Sen},\ and\ \citenamefont
  {{\.{Z}}ukowski}}]{SenDeSZ02}%
  \BibitemOpen
  \bibfield  {author} {\bibinfo {author} {\bibfnamefont {A.}~\bibnamefont
  {Sen(De)}}, \bibinfo {author} {\bibfnamefont {U.}~\bibnamefont {Sen}}, \ and\
  \bibinfo {author} {\bibfnamefont {M.}~\bibnamefont {{\.{Z}}ukowski}},\
  }\bibfield  {title} {\enquote {\bibinfo {title} {{Functional Bell
  inequalities can serve as a stronger entanglement witness than conventional
  Bell inequalities}},}\ }\href {\doibase 10.1103/PhysRevA.66.062318}
  {\bibfield  {journal} {\bibinfo  {journal} {Physical Review A - Atomic,
  Molecular, and Optical Physics}\ }\textbf {\bibinfo {volume} {66}},\ \bibinfo
  {pages} {4} (\bibinfo {year} {2002})}\BibitemShut {NoStop}%
\bibitem [{\citenamefont {Cirel'son}(1980)}]{Cirelson80}%
  \BibitemOpen
  \bibfield  {author} {\bibinfo {author} {\bibfnamefont {B.~S.}\ \bibnamefont
  {Cirel'son}},\ }\bibfield  {title} {\enquote {\bibinfo {title} {{Quantum
  generalizations of Bell's inequality}},}\ }\href {\doibase
  10.1007/BF00417500} {\bibfield  {journal} {\bibinfo  {journal} {Letters in
  Mathematical Physics}\ }\textbf {\bibinfo {volume} {4}},\ \bibinfo {pages}
  {93--100} (\bibinfo {year} {1980})}\BibitemShut {NoStop}%
\bibitem [{\citenamefont {Hirsch}\ \emph {et~al.}(2017)\citenamefont {Hirsch},
  \citenamefont {Quintino}, \citenamefont {V\'{e}rtesi}, \citenamefont
  {Navascu\'{e}s},\ and\ \citenamefont {Brunner}}]{HirschQVNB17}%
  \BibitemOpen
  \bibfield  {author} {\bibinfo {author} {\bibfnamefont {F.}~\bibnamefont
  {Hirsch}}, \bibinfo {author} {\bibfnamefont {M.~T.}\ \bibnamefont
  {Quintino}}, \bibinfo {author} {\bibfnamefont {T.}~\bibnamefont
  {V\'{e}rtesi}}, \bibinfo {author} {\bibfnamefont {M.}~\bibnamefont
  {Navascu\'{e}s}}, \ and\ \bibinfo {author} {\bibfnamefont {N.}~\bibnamefont
  {Brunner}},\ }\bibfield  {title} {\enquote {\bibinfo {title} {{Better local
  hidden variable models for two-bit Werner states and an upper bound on the
  Grothendieck constant $K_G(3)$}},}\ }\href {\doibase 10.22331/q-2017-04-25-3}
  {\bibfield  {journal} {\bibinfo  {journal} {Quantum}\ }\textbf {\bibinfo
  {volume} {1}},\ \bibinfo {pages} {3} (\bibinfo {year} {2017})}\BibitemShut
  {NoStop}%
\bibitem [{\citenamefont {Nagata}\ \emph {et~al.}(2004)\citenamefont {Nagata},
  \citenamefont {Laskowski}, \citenamefont {Wie{\'{s}}niak},\ and\
  \citenamefont {{\.{Z}}ukowski}}]{NagataLWZ04}%
  \BibitemOpen
  \bibfield  {author} {\bibinfo {author} {\bibfnamefont {K.}~\bibnamefont
  {Nagata}}, \bibinfo {author} {\bibfnamefont {W.}~\bibnamefont {Laskowski}},
  \bibinfo {author} {\bibfnamefont {M.}~\bibnamefont {Wie{\'{s}}niak}}, \ and\
  \bibinfo {author} {\bibfnamefont {M.}~\bibnamefont {{\.{Z}}ukowski}},\
  }\bibfield  {title} {\enquote {\bibinfo {title} {{Rotational invariance as an
  additional constraint on local realism}},}\ }\href {\doibase
  10.1103/PhysRevLett.93.230403} {\bibfield  {journal} {\bibinfo  {journal}
  {Physical Review Letters}\ }\textbf {\bibinfo {volume} {93}} (\bibinfo {year}
  {2004}),\ 10.1103/PhysRevLett.93.230403}\BibitemShut {NoStop}%
\bibitem [{\citenamefont {Braunstein}\ and\ \citenamefont
  {Caves}(1990)}]{BraunsteinC90}%
  \BibitemOpen
  \bibfield  {author} {\bibinfo {author} {\bibfnamefont {S.~L.}\ \bibnamefont
  {Braunstein}}\ and\ \bibinfo {author} {\bibfnamefont {C.~M.}\ \bibnamefont
  {Caves}},\ }\bibfield  {title} {\enquote {\bibinfo {title} {{Wringing out
  better Bell inequalities}},}\ }\href {\doibase 10.1016/0003-4916(90)90339-P}
  {\bibfield  {journal} {\bibinfo  {journal} {Annals of Physics}\ }\textbf
  {\bibinfo {volume} {202}},\ \bibinfo {pages} {22--56} (\bibinfo {year}
  {1990})}\BibitemShut {NoStop}%
\bibitem [{\citenamefont {Wang}\ \emph {et~al.}(2017)\citenamefont {Wang},
  \citenamefont {Singh},\ and\ \citenamefont {Navascu{\'{e}}s}}]{WangSN17}%
  \BibitemOpen
  \bibfield  {author} {\bibinfo {author} {\bibfnamefont {Z.}~\bibnamefont
  {Wang}}, \bibinfo {author} {\bibfnamefont {S.}~\bibnamefont {Singh}}, \ and\
  \bibinfo {author} {\bibfnamefont {M.}~\bibnamefont {Navascu{\'{e}}s}},\
  }\bibfield  {title} {\enquote {\bibinfo {title} {{Entanglement and
  Nonlocality in Infinite 1D Systems}},}\ }\href {\doibase
  10.1103/PhysRevLett.118.230401} {\bibfield  {journal} {\bibinfo  {journal}
  {Physical Review Letters}\ }\textbf {\bibinfo {volume} {118}},\ \bibinfo
  {pages} {230401} (\bibinfo {year} {2017})}\BibitemShut {NoStop}%
\bibitem [{\citenamefont {Khalfin}\ and\ \citenamefont
  {Tsirelson}(1985)}]{KhalfinT85}%
  \BibitemOpen
  \bibfield  {author} {\bibinfo {author} {\bibfnamefont {L.~A.}\ \bibnamefont
  {Khalfin}}\ and\ \bibinfo {author} {\bibfnamefont {B.~S.}\ \bibnamefont
  {Tsirelson}},\ }\bibfield  {title} {\enquote {\bibinfo {title} {{Quantum and
  Quasi-classical Analogs Of Bell Inequalities}},}\ }in\ \href@noop {} {\emph
  {\bibinfo {booktitle} {Symposium on the foundations of modern physics}}},\
  \bibinfo {editor} {edited by\ \bibinfo {editor} {\bibfnamefont
  {P.}~\bibnamefont {Lahti}}\ and\ \bibinfo {editor} {\bibfnamefont
  {P.}~\bibnamefont {Mittelstaedt}}}\ (\bibinfo  {publisher} {World Scientific
  Publishing Co.},\ \bibinfo {year} {1985})\ pp.\ \bibinfo {pages}
  {441--460}\BibitemShut {NoStop}%
\bibitem [{\citenamefont {Paw{\l}owski}\ \emph {et~al.}(2009)\citenamefont
  {Paw{\l}owski}, \citenamefont {Paterek}, \citenamefont {Kaszlikowski},
  \citenamefont {Scarani}, \citenamefont {Winter},\ and\ \citenamefont
  {Zukowski}}]{PawlowskiPKSWZ09}%
  \BibitemOpen
  \bibfield  {author} {\bibinfo {author} {\bibfnamefont {M.}~\bibnamefont
  {Paw{\l}owski}}, \bibinfo {author} {\bibfnamefont {T.}~\bibnamefont
  {Paterek}}, \bibinfo {author} {\bibfnamefont {D.}~\bibnamefont
  {Kaszlikowski}}, \bibinfo {author} {\bibfnamefont {V.}~\bibnamefont
  {Scarani}}, \bibinfo {author} {\bibfnamefont {A.}~\bibnamefont {Winter}}, \
  and\ \bibinfo {author} {\bibfnamefont {M.}~\bibnamefont {Zukowski}},\
  }\bibfield  {title} {\enquote {\bibinfo {title} {{Information causality as a
  physical principle}},}\ }\href {\doibase 10.1038/nature08400} {\bibfield
  {journal} {\bibinfo  {journal} {Nature}\ }\textbf {\bibinfo {volume} {461}},\
  \bibinfo {pages} {1101--1104} (\bibinfo {year} {2009})}\BibitemShut {NoStop}%
\bibitem [{\citenamefont {Yang}\ \emph {et~al.}(2011)\citenamefont {Yang},
  \citenamefont {Navascu{\'{e}}s}, \citenamefont {Sheridan},\ and\
  \citenamefont {Scarani}}]{YangNSS11}%
  \BibitemOpen
  \bibfield  {author} {\bibinfo {author} {\bibfnamefont {T.~H.}\ \bibnamefont
  {Yang}}, \bibinfo {author} {\bibfnamefont {M.}~\bibnamefont
  {Navascu{\'{e}}s}}, \bibinfo {author} {\bibfnamefont {L.}~\bibnamefont
  {Sheridan}}, \ and\ \bibinfo {author} {\bibfnamefont {V.}~\bibnamefont
  {Scarani}},\ }\bibfield  {title} {\enquote {\bibinfo {title} {{Quantum Bell
  inequalities from macroscopic locality}},}\ }\href {\doibase
  10.1103/PhysRevA.83.022105} {\bibfield  {journal} {\bibinfo  {journal}
  {Physical Review A}\ }\textbf {\bibinfo {volume} {83}},\ \bibinfo {pages}
  {022105} (\bibinfo {year} {2011})}\BibitemShut {NoStop}%
\bibitem [{\citenamefont {Brassard}\ \emph {et~al.}(2006)\citenamefont
  {Brassard}, \citenamefont {Buhrman}, \citenamefont {Linden}, \citenamefont
  {M{\'{e}}thot}, \citenamefont {Tapp},\ and\ \citenamefont
  {Unger}}]{BrassardBLMTU06}%
  \BibitemOpen
  \bibfield  {author} {\bibinfo {author} {\bibfnamefont {G.}~\bibnamefont
  {Brassard}}, \bibinfo {author} {\bibfnamefont {H.}~\bibnamefont {Buhrman}},
  \bibinfo {author} {\bibfnamefont {N.}~\bibnamefont {Linden}}, \bibinfo
  {author} {\bibfnamefont {A.}~\bibnamefont {M{\'{e}}thot}}, \bibinfo {author}
  {\bibfnamefont {A.}~\bibnamefont {Tapp}}, \ and\ \bibinfo {author}
  {\bibfnamefont {F.}~\bibnamefont {Unger}},\ }\bibfield  {title} {\enquote
  {\bibinfo {title} {{Limit on Nonlocality in Any World in Which Communication
  Complexity Is Not Trivial}},}\ }\href {\doibase
  10.1103/PhysRevLett.96.250401} {\bibfield  {journal} {\bibinfo  {journal}
  {Physical Review Letters}\ }\textbf {\bibinfo {volume} {96}},\ \bibinfo
  {pages} {250401} (\bibinfo {year} {2006})}\BibitemShut {NoStop}%
\bibitem [{\citenamefont {Navascu{\'{e}}s}\ \emph {et~al.}(2015)\citenamefont
  {Navascu{\'{e}}s}, \citenamefont {Guryanova}, \citenamefont {Hoban},\ and\
  \citenamefont {Ac{\'{i}}n}}]{NavascuesGHA15}%
  \BibitemOpen
  \bibfield  {author} {\bibinfo {author} {\bibfnamefont {M.}~\bibnamefont
  {Navascu{\'{e}}s}}, \bibinfo {author} {\bibfnamefont {Y.}~\bibnamefont
  {Guryanova}}, \bibinfo {author} {\bibfnamefont {M.~J.}\ \bibnamefont
  {Hoban}}, \ and\ \bibinfo {author} {\bibfnamefont {A.}~\bibnamefont
  {Ac{\'{i}}n}},\ }\bibfield  {title} {\enquote {\bibinfo {title} {{Almost
  quantum correlations}},}\ }\href {\doibase 10.1038/ncomms7288} {\bibfield
  {journal} {\bibinfo  {journal} {Nature Communications}\ }\textbf {\bibinfo
  {volume} {6}},\ \bibinfo {pages} {6288} (\bibinfo {year} {2015})}\BibitemShut
  {NoStop}%
\bibitem [{\citenamefont {P{\'{a}}l}\ and\ \citenamefont
  {V{\'{e}}rtesi}(2009)}]{PalV09}%
  \BibitemOpen
  \bibfield  {author} {\bibinfo {author} {\bibfnamefont {K.~F.}\ \bibnamefont
  {P{\'{a}}l}}\ and\ \bibinfo {author} {\bibfnamefont {T.}~\bibnamefont
  {V{\'{e}}rtesi}},\ }\bibfield  {title} {\enquote {\bibinfo {title}
  {{Concavity of the set of quantum probabilities for any given dimension}},}\
  }\href {\doibase 10.1103/PhysRevA.80.042114} {\bibfield  {journal} {\bibinfo
  {journal} {Physical Review A}\ }\textbf {\bibinfo {volume} {80}},\ \bibinfo
  {pages} {042114} (\bibinfo {year} {2009})}\BibitemShut {NoStop}%
\bibitem [{\citenamefont {Donohue}\ and\ \citenamefont
  {Wolfe}(2015)}]{DonohueW15}%
  \BibitemOpen
  \bibfield  {author} {\bibinfo {author} {\bibfnamefont {J.~M.}\ \bibnamefont
  {Donohue}}\ and\ \bibinfo {author} {\bibfnamefont {E.}~\bibnamefont
  {Wolfe}},\ }\bibfield  {title} {\enquote {\bibinfo {title} {{Identifying
  nonconvexity in the sets of limited-dimension quantum correlations}},}\
  }\href {\doibase 10.1103/PhysRevA.92.062120} {\bibfield  {journal} {\bibinfo
  {journal} {Physical Review A}\ }\textbf {\bibinfo {volume} {92}},\ \bibinfo
  {pages} {062120} (\bibinfo {year} {2015})}\BibitemShut {NoStop}%
\bibitem [{\citenamefont {Ac{\'{i}}n}\ \emph {et~al.}(2010)\citenamefont
  {Ac{\'{i}}n}, \citenamefont {Augusiak}, \citenamefont {Cavalcanti},
  \citenamefont {Hadley}, \citenamefont {Korbicz}, \citenamefont {Lewenstein},
  \citenamefont {Masanes},\ and\ \citenamefont {Piani}}]{AcinACHKLMP10}%
  \BibitemOpen
  \bibfield  {author} {\bibinfo {author} {\bibfnamefont {A.}~\bibnamefont
  {Ac{\'{i}}n}}, \bibinfo {author} {\bibfnamefont {R.}~\bibnamefont
  {Augusiak}}, \bibinfo {author} {\bibfnamefont {D.}~\bibnamefont
  {Cavalcanti}}, \bibinfo {author} {\bibfnamefont {C.}~\bibnamefont {Hadley}},
  \bibinfo {author} {\bibfnamefont {J.~K.}\ \bibnamefont {Korbicz}}, \bibinfo
  {author} {\bibfnamefont {M.}~\bibnamefont {Lewenstein}}, \bibinfo {author}
  {\bibfnamefont {Ll.}\ \bibnamefont {Masanes}}, \ and\ \bibinfo {author}
  {\bibfnamefont {M.}~\bibnamefont {Piani}},\ }\bibfield  {title} {\enquote
  {\bibinfo {title} {{Unified Framework for Correlations in Terms of Local
  Quantum Observables}},}\ }\href {\doibase 10.1103/PhysRevLett.104.140404}
  {\bibfield  {journal} {\bibinfo  {journal} {Physical Review Letters}\
  }\textbf {\bibinfo {volume} {104}},\ \bibinfo {pages} {140404} (\bibinfo
  {year} {2010})}\BibitemShut {NoStop}%
\bibitem [{\citenamefont {Wootters}(1980)}]{Wootters80}%
  \BibitemOpen
  \bibfield  {author} {\bibinfo {author} {\bibfnamefont {W.~K.}\ \bibnamefont
  {Wootters}},\ }\emph {\bibinfo {title} {The acquisition of information from
  quantum measurements}},\ \href@noop {} {Ph.D. thesis},\ \bibinfo  {school}
  {University of Texas at Austin} (\bibinfo {year} {1980})\BibitemShut
  {NoStop}%
\bibitem [{\citenamefont {M{\"{u}}ller}\ and\ \citenamefont
  {Masanes}(2013)}]{MuellerM13}%
  \BibitemOpen
  \bibfield  {author} {\bibinfo {author} {\bibfnamefont {M.~P.}\ \bibnamefont
  {M{\"{u}}ller}}\ and\ \bibinfo {author} {\bibfnamefont {Ll.}\ \bibnamefont
  {Masanes}},\ }\bibfield  {title} {\enquote {\bibinfo {title}
  {{Three-dimensionality of space and the quantum bit: an information-theoretic
  approach}},}\ }\href {\doibase 10.1088/1367-2630/15/5/053040} {\bibfield
  {journal} {\bibinfo  {journal} {New Journal of Physics}\ }\textbf {\bibinfo
  {volume} {15}},\ \bibinfo {pages} {053040} (\bibinfo {year}
  {2013})}\BibitemShut {NoStop}%
\bibitem [{\citenamefont {H{\"{o}}hn}\ and\ \citenamefont
  {M{\"{u}}ller}(2016)}]{HoehnM16}%
  \BibitemOpen
  \bibfield  {author} {\bibinfo {author} {\bibfnamefont {P.~A.}\ \bibnamefont
  {H{\"{o}}hn}}\ and\ \bibinfo {author} {\bibfnamefont {M.~P.}\ \bibnamefont
  {M{\"{u}}ller}},\ }\bibfield  {title} {\enquote {\bibinfo {title} {{An
  operational approach to spacetime symmetries: Lorentz transformations from
  quantum communication}},}\ }\href {\doibase 10.1088/1367-2630/18/6/063026}
  {\bibfield  {journal} {\bibinfo  {journal} {New Journal of Physics}\ }\textbf
  {\bibinfo {volume} {18}},\ \bibinfo {pages} {063026} (\bibinfo {year}
  {2016})}\BibitemShut {NoStop}%
\bibitem [{\citenamefont {Garner}\ \emph {et~al.}(2017)\citenamefont {Garner},
  \citenamefont {M{\"{u}}ller},\ and\ \citenamefont {Dahlsten}}]{GarnerMD17}%
  \BibitemOpen
  \bibfield  {author} {\bibinfo {author} {\bibfnamefont {A.~J.~P.}\
  \bibnamefont {Garner}}, \bibinfo {author} {\bibfnamefont {M.~P.}\
  \bibnamefont {M{\"{u}}ller}}, \ and\ \bibinfo {author} {\bibfnamefont
  {O.~C.~O.}\ \bibnamefont {Dahlsten}},\ }\bibfield  {title} {\enquote
  {\bibinfo {title} {{The complex and quaternionic quantum bit from relativity
  of simultaneity on an interferometer}},}\ }\href {\doibase
  10.1098/rspa.2017.0596} {\bibfield  {journal} {\bibinfo  {journal}
  {Proceedings of the Royal Society A}\ }\textbf {\bibinfo {volume} {473}},\
  \bibinfo {pages} {20170596} (\bibinfo {year} {2017})}\BibitemShut {NoStop}%
\bibitem [{\citenamefont {Sepanski}(2007)}]{Sepanski07}%
  \BibitemOpen
  \bibfield  {author} {\bibinfo {author} {\bibfnamefont {M.~R.}\ \bibnamefont
  {Sepanski}},\ }\href@noop {} {\emph {\bibinfo {title} {{Complex Lie
  groups}}}}\ (\bibinfo  {publisher} {Springer},\ \bibinfo {address} {New
  York},\ \bibinfo {year} {2007})\BibitemShut {NoStop}%
\bibitem [{\citenamefont {Branford}\ \emph {et~al.}(2018)\citenamefont
  {Branford}, \citenamefont {Dahlsten},\ and\ \citenamefont
  {Garner}}]{BranfordDG18}%
  \BibitemOpen
  \bibfield  {author} {\bibinfo {author} {\bibfnamefont {D.}~\bibnamefont
  {Branford}}, \bibinfo {author} {\bibfnamefont {O.~C.~O.}\ \bibnamefont
  {Dahlsten}}, \ and\ \bibinfo {author} {\bibfnamefont {A.~J.~P.}\ \bibnamefont
  {Garner}},\ }\bibfield  {title} {\enquote {\bibinfo {title} {{On Defining the
  Hamiltonian Beyond Quantum Theory}},}\ }\href {\doibase
  10.1007/s10701-018-0205-9} {\bibfield  {journal} {\bibinfo  {journal}
  {Foundations of Physics}\ }\textbf {\bibinfo {volume} {48}},\ \bibinfo
  {pages} {982--1006} (\bibinfo {year} {2018})}\BibitemShut {NoStop}%
\bibitem [{\citenamefont {van Himbeeck}\ \emph {et~al.}(2017)\citenamefont {van
  Himbeeck}, \citenamefont {Woodhead}, \citenamefont {Cerf}, \citenamefont
  {Garc{\'{i}}a-Patr{\'{o}}n},\ and\ \citenamefont
  {Pironio}}]{vanHimbeeckWCGP17}%
  \BibitemOpen
  \bibfield  {author} {\bibinfo {author} {\bibfnamefont {T.}~\bibnamefont {van
  Himbeeck}}, \bibinfo {author} {\bibfnamefont {E.}~\bibnamefont {Woodhead}},
  \bibinfo {author} {\bibfnamefont {N.~J.}\ \bibnamefont {Cerf}}, \bibinfo
  {author} {\bibfnamefont {R.}~\bibnamefont {Garc{\'{i}}a-Patr{\'{o}}n}}, \
  and\ \bibinfo {author} {\bibfnamefont {S.}~\bibnamefont {Pironio}},\
  }\bibfield  {title} {\enquote {\bibinfo {title} {{Semi-device-independent
  framework based on natural physical assumptions}},}\ }\href {\doibase
  10.22331/q-2017-11-18-33} {\bibfield  {journal} {\bibinfo  {journal}
  {Quantum}\ }\textbf {\bibinfo {volume} {1}},\ \bibinfo {pages} {33} (\bibinfo
  {year} {2017})}\BibitemShut {NoStop}%
\bibitem [{\citenamefont {{Van Himbeeck}}\ and\ \citenamefont
  {Pironio}(2019)}]{VanHimbeeckP19}%
  \BibitemOpen
  \bibfield  {author} {\bibinfo {author} {\bibfnamefont {T.}~\bibnamefont {{Van
  Himbeeck}}}\ and\ \bibinfo {author} {\bibfnamefont {S.}~\bibnamefont
  {Pironio}},\ }\bibfield  {title} {\enquote {\bibinfo {title} {{Correlations
  and randomness generation based on energy constraints}},}\ }\href
  {http://arxiv.org/abs/1905.09117} {\bibfield  {journal} {\bibinfo  {journal}
  {{arXiv:1905.09117}}\ } (\bibinfo {year} {2019})}\BibitemShut {NoStop}%
\bibitem [{\citenamefont {Rusca}\ \emph {et~al.}(2019)\citenamefont {Rusca},
  \citenamefont {van Himbeeck}, \citenamefont {Martin}, \citenamefont {Brask},
  \citenamefont {Shi}, \citenamefont {Pironio}, \citenamefont {Brunner},\ and\
  \citenamefont {Zbinden}}]{Rusca19}%
  \BibitemOpen
  \bibfield  {author} {\bibinfo {author} {\bibfnamefont {D.}~\bibnamefont
  {Rusca}}, \bibinfo {author} {\bibfnamefont {T}~\bibnamefont {van Himbeeck}},
  \bibinfo {author} {\bibfnamefont {A.}~\bibnamefont {Martin}}, \bibinfo
  {author} {\bibfnamefont {J.~B.}\ \bibnamefont {Brask}}, \bibinfo {author}
  {\bibfnamefont {W.}~\bibnamefont {Shi}}, \bibinfo {author} {\bibfnamefont
  {S.}~\bibnamefont {Pironio}}, \bibinfo {author} {\bibfnamefont
  {N.}~\bibnamefont {Brunner}}, \ and\ \bibinfo {author} {\bibfnamefont
  {H.}~\bibnamefont {Zbinden}},\ }\bibfield  {title} {\enquote {\bibinfo
  {title} {{Practical self-testing quantum random number generator based on an
  energy bound}},}\ }\href {http://arxiv.org/abs/1904.04819} {\bibfield
  {journal} {\bibinfo  {journal} {{arXiv:1904.04819}}\ } (\bibinfo {year}
  {2019})}\BibitemShut {NoStop}%
\bibitem [{\citenamefont {Sorkin}(1994)}]{Sorkin94}%
  \BibitemOpen
  \bibfield  {author} {\bibinfo {author} {\bibfnamefont {R.~D.}\ \bibnamefont
  {Sorkin}},\ }\bibfield  {title} {\enquote {\bibinfo {title} {{Quantum
  Mechanics as Quantum Measure Theory}},}\ }\href {\doibase
  10.1142/S021773239400294X} {\bibfield  {journal} {\bibinfo  {journal} {Modern
  Physics Letters A}\ }\textbf {\bibinfo {volume} {09}},\ \bibinfo {pages}
  {3119--3127} (\bibinfo {year} {1994})}\BibitemShut {NoStop}%
\bibitem [{\citenamefont {Ududec}\ \emph {et~al.}(2010)\citenamefont {Ududec},
  \citenamefont {Barnum},\ and\ \citenamefont {Emerson}}]{UdudecBE10}%
  \BibitemOpen
  \bibfield  {author} {\bibinfo {author} {\bibfnamefont {C.}~\bibnamefont
  {Ududec}}, \bibinfo {author} {\bibfnamefont {H.}~\bibnamefont {Barnum}}, \
  and\ \bibinfo {author} {\bibfnamefont {J.}~\bibnamefont {Emerson}},\
  }\bibfield  {title} {\enquote {\bibinfo {title} {{Three Slit Experiments and
  the Structure of Quantum Theory}},}\ }\href {\doibase
  10.1007/s10701-010-9429-z} {\bibfield  {journal} {\bibinfo  {journal}
  {Foundations of Physics}\ }\textbf {\bibinfo {volume} {41}},\ \bibinfo
  {pages} {396--405} (\bibinfo {year} {2010})}\BibitemShut {NoStop}%
\bibitem [{\citenamefont {Sinha}\ \emph {et~al.}(2010)\citenamefont {Sinha},
  \citenamefont {Couteau}, \citenamefont {Jennewein}, \citenamefont
  {Laflamme},\ and\ \citenamefont {Weihs}}]{SinhaCJLW10}%
  \BibitemOpen
  \bibfield  {author} {\bibinfo {author} {\bibfnamefont {U.}~\bibnamefont
  {Sinha}}, \bibinfo {author} {\bibfnamefont {C.}~\bibnamefont {Couteau}},
  \bibinfo {author} {\bibfnamefont {T.}~\bibnamefont {Jennewein}}, \bibinfo
  {author} {\bibfnamefont {R.}~\bibnamefont {Laflamme}}, \ and\ \bibinfo
  {author} {\bibfnamefont {G.}~\bibnamefont {Weihs}},\ }\bibfield  {title}
  {\enquote {\bibinfo {title} {{Ruling out multi-order interference in quantum
  mechanics.}}}\ }\href {\doibase 10.1126/science.1190545} {\bibfield
  {journal} {\bibinfo  {journal} {Science (New York, N.Y.)}\ }\textbf {\bibinfo
  {volume} {329}},\ \bibinfo {pages} {418--21} (\bibinfo {year}
  {2010})}\BibitemShut {NoStop}%
\bibitem [{Fin(2016)}]{FiniteDim}%
  \BibitemOpen
  \href
  {{http://www.encyclopediaofmath.org/index.php?title=Finite-dimensional_representation&oldid=39848}}
  {\enquote {\bibinfo {title} {Finite-dimensional representation.}}\ }\bibinfo
  {howpublished} {Encyclopedia of Mathematics.} (\bibinfo {year}
  {2016})\BibitemShut {NoStop}%
\bibitem [{\citenamefont {Molchanov}\ \emph {et~al.}(1995)\citenamefont
  {Molchanov}, \citenamefont {Klimyk},\ and\ \citenamefont
  {Vilenkin}}]{Kirillov95}%
  \BibitemOpen
  \bibfield  {author} {\bibinfo {author} {\bibfnamefont {V.~F.}\ \bibnamefont
  {Molchanov}}, \bibinfo {author} {\bibfnamefont {A.~U.}\ \bibnamefont
  {Klimyk}}, \ and\ \bibinfo {author} {\bibfnamefont {N.~Ya.}\ \bibnamefont
  {Vilenkin}},\ }\href@noop {} {\emph {\bibinfo {title} {{Representation Theory
  and Noncommutative Harmonic Analysis II}}}},\ edited by\ \bibinfo {editor}
  {\bibfnamefont {A.~A.}\ \bibnamefont {Kirillov}}\ (\bibinfo  {publisher}
  {Springer Verlag},\ \bibinfo {year} {1995})\BibitemShut {NoStop}%
\bibitem [{\citenamefont {Kleinmann}\ \emph {et~al.}(2013)\citenamefont
  {Kleinmann}, \citenamefont {Osborne}, \citenamefont {Scholz},\ and\
  \citenamefont {Werner}}]{KleinmannOSW13}%
  \BibitemOpen
  \bibfield  {author} {\bibinfo {author} {\bibfnamefont {M.}~\bibnamefont
  {Kleinmann}}, \bibinfo {author} {\bibfnamefont {T.~J.}\ \bibnamefont
  {Osborne}}, \bibinfo {author} {\bibfnamefont {V.~B.}\ \bibnamefont {Scholz}},
  \ and\ \bibinfo {author} {\bibfnamefont {A.~H.}\ \bibnamefont {Werner}},\
  }\bibfield  {title} {\enquote {\bibinfo {title} {{Typical Local Measurements
  in Generalized Probabilistic Theories: Emergence of Quantum Bipartite
  Correlations}},}\ }\href {\doibase 10.1103/PhysRevLett.110.040403} {\bibfield
   {journal} {\bibinfo  {journal} {Physical Review Letters}\ }\textbf {\bibinfo
  {volume} {110}},\ \bibinfo {pages} {040403} (\bibinfo {year}
  {2013})}\BibitemShut {NoStop}%
\bibitem [{\citenamefont {Barnum}\ \emph {et~al.}(2010)\citenamefont {Barnum},
  \citenamefont {Beigi}, \citenamefont {Boixo}, \citenamefont {Elliott},\ and\
  \citenamefont {Wehner}}]{BarnumBBEW10}%
  \BibitemOpen
  \bibfield  {author} {\bibinfo {author} {\bibfnamefont {H.}~\bibnamefont
  {Barnum}}, \bibinfo {author} {\bibfnamefont {S.}~\bibnamefont {Beigi}},
  \bibinfo {author} {\bibfnamefont {S.}~\bibnamefont {Boixo}}, \bibinfo
  {author} {\bibfnamefont {M.~B.}\ \bibnamefont {Elliott}}, \ and\ \bibinfo
  {author} {\bibfnamefont {S.}~\bibnamefont {Wehner}},\ }\bibfield  {title}
  {\enquote {\bibinfo {title} {{Local Quantum Measurement and No-Signaling
  Imply Quantum Correlations}},}\ }\href {\doibase
  10.1103/PhysRevLett.104.140401} {\bibfield  {journal} {\bibinfo  {journal}
  {Phys. Rev. Lett.}\ }\textbf {\bibinfo {volume} {104}},\ \bibinfo {pages}
  {140401} (\bibinfo {year} {2010})}\BibitemShut {NoStop}%
\bibitem [{\citenamefont {Masanes}(2005)}]{Masanes05}%
  \BibitemOpen
  \bibfield  {author} {\bibinfo {author} {\bibfnamefont {Ll.}\ \bibnamefont
  {Masanes}},\ }\bibfield  {title} {\enquote {\bibinfo {title} {{Extremal
  quantum correlations for N parties with two dichotomic observables per
  site}},}\ }\href {http://arxiv.org/abs/quant-ph/0512100} {\bibfield
  {journal} {\bibinfo  {journal} {{arXiv:quant-ph/0512100}}\ } (\bibinfo {year}
  {2005})}\BibitemShut {NoStop}%
\bibitem [{\citenamefont {Toner}\ and\ \citenamefont
  {Verstraete}(2006)}]{TonerV06}%
  \BibitemOpen
  \bibfield  {author} {\bibinfo {author} {\bibfnamefont {B.}~\bibnamefont
  {Toner}}\ and\ \bibinfo {author} {\bibfnamefont {F.}~\bibnamefont
  {Verstraete}},\ }\bibfield  {title} {\enquote {\bibinfo {title} {{Monogamy of
  Bell correlations and Tsirelson's bound}},}\ }\href
  {http://arxiv.org/abs/quant-ph/0611001} {\bibfield  {journal} {\bibinfo
  {journal} {{arXiv:quant-ph/0611001}}\ } (\bibinfo {year} {2006})}\BibitemShut
  {NoStop}%
\end{thebibliography}
\end{document}